\documentclass[amsmath,twocolumn,superscriptaddress,prx,aps]{revtex4-1}
\usepackage{amsthm,amsfonts,graphicx, color}
\usepackage{bm}
\usepackage{graphicx}
\usepackage{braket}
\usepackage{bbm}
\usepackage{subfigure}
\usepackage{amssymb}
\usepackage[colorlinks,urlcolor=blue,linkcolor=blue,anchorcolor=blue,citecolor=blue]{hyperref}
\usepackage{wasysym}
\usepackage{color}
\stepcounter{secnumdepth}
\stepcounter{tocdepth}
\usepackage[utf8]{inputenc}
\usepackage{ulem}
\usepackage{enumitem}

\newcommand{\be}{\begin{equation}}
\newcommand{\ee}{\end{equation}}
\newcommand{\bea}{\begin{eqnarray}}
\newcommand{\eea}{\end{eqnarray}}

\newcommand{\tr}{{\rm tr\/}\,}

\newcommand{\Tr}{{\rm Tr}\,}
\renewcommand{\epsilon}{\varepsilon}
\renewcommand{\vec}[1]{{\bf #1}}

\newtheorem{lemma}{Lemma}

\newcommand{\T}{\mathcal{T}}

\newcommand{\Hi}{\mathcal{H}}
\newcommand{\outerp}[2]{\ket{#1}\!\bra{#2}}

\newcommand{\pj}{\mathcal{P}}
\newcommand{\up}[2]{#1^{(#2)}}

\newcommand{\supp}{\text{Supp}}
\newcommand{\cut}{\text{Cut}}
\newcommand{\mc}{\mathcal}
\newcommand{\mf}{\mathfrak}
\newcommand{\mb}{\mathbb}
\newcommand{\spn}{\text{Span}}

\def\beq{\begin{equation}}
\def\eeq{\end{equation}}
\def\bea{\begin{eqnarray}}
\def\eea{\end{eqnarray}}

\begin{document}
\title{Recoverable Information and Emergent Conservation Laws in Fracton Stabilizer Codes}

\author{A.T. Schmitz}\email{albert.schmitz@colorado.edu}
\affiliation{Department of Physics, University of Colorado at Boulder, Boulder CO 80309, USA}
\author{Han Ma}\email{han.ma@colorado.edu}
\affiliation{Department of Physics, University of Colorado at Boulder, Boulder CO 80309, USA}
\author{Rahul M. Nandkishore}\email{rahul.nandkishore@colorado.edu}
\affiliation{Department of Physics, University of Colorado at Boulder, Boulder CO 80309, USA}
\affiliation{Center for Theory of Quantum Matter, University of Colorado at Boulder, Boulder CO 80309, USA}
\author{S. A. Parameswaran}\email{sid.parameswaran@physics.ox.ac.uk}
\affiliation{The Rudolf Peierls Centre for Theoretical Physics, University of Oxford, Oxford OX1 3NP, UK}

\begin{abstract}
We introduce a new quantity, that we term {\it recoverable information}, defined for stabilizer Hamiltonians.  
For such models, the recoverable information provides a measure of the topological information, as well as a physical interpretation, which is complementary to topological entanglement entropy. 
We discuss three different ways to calculate the recoverable information, and prove their equivalence. To demonstrate its utility, 
 we compute recoverable information for  {\it fracton models} 
 using all three methods where appropriate. From the recoverable information, we deduce the existence of emergent $Z_2$ Gauss-law type constraints, which  in turn imply emergent $Z_2$ conservation laws for point-like quasiparticle excitations of an underlying topologically ordered phase.  
\end{abstract}
\maketitle
 
\tableofcontents
 
\section{Introduction}\label{sec:intro}
Ideas from quantum information have provided powerful tools for the study of complex quantum many-body systems and the myriad phases of matter they host. Quantum entanglement has been a uniquely influential concept in this regard --- specifically, its characterization in terms of {entanglement entropy}. While discussions of entanglement entropy pervade many areas of modern theoretical physics, it has proven particularly potent in the characterization of topological order (see, e.g., Ref.~\onlinecite{GroverReview} for a review). For two-dimensional gapped topologically ordered phases, the scaling of the entanglement entropy of a subsystem with its size contains a universal subleading `constant' term~\cite{KitaevPreskill, LevinWen}. This is intimately related to the topological quantum field theory (TQFT) description of the topological phase and provides a partial characterization of the nature of topological order. However, while two dimensional topological orders are relatively well understood, the book of topological order in higher dimensions is still being written. 

Indeed, the discovery of three-dimensional {\it fracton} models~\cite{Chamon2005,Bravyi2011,Haah2011,Yoshida2013,Vijay2015,Vijay2016,Nussinov2009,Nussinov2009(2)} has begun a new chapter in this book. These models exhibit many of the characteristic properties of topological order, such as locally indistinguishable degenerate ground states on manifolds of non-trivial topology, and local excitations that cannot be created by purely local operators. However, unlike more familiar topologically ordered states, the point-like excitations of these models are restricted to move only within lower-dimensional subspaces such as along a line or in a plane. This leads to a sub-extensive ground state degeneracy on manifolds of nontrivial topology, rather than the {\it finite} ground state degeneracy expected of topological orders described by emergent discrete gauge theories.
As such, fracton models have begun to draw intensive interest~\cite{Vijay2016,Williamson2016, Pretko1, Pretko2, Vijay2017,Hermele, Prem2017, Kim2017, Hsieh2017, Pretko2d, Slagle2017, Devakul, PetrovaRegnault, CMfractons}. Given the insights previously gleaned into topological order by the consideration of topological entanglement, an investigation of the topological entanglement of fracton models is both timely and interesting.  

Two recent works have made important progress in this direction. Ref.~\onlinecite{Lu2017} places a linear-scaling {bound} on the dependence of the non-local entanglement on subsystem size in  several fracton models, thereby establishing that the fracton models considered do possess non-local entanglement. Meanwhile in a recent paper \cite{companion1} we have explicitly computed the topological entanglement entropy of the same archetypal fracton models presented here. Our results are consistent with the bounds of Ref.~\onlinecite{Lu2017}, and  match those from a distinct approach that uses a tensor-network representation of the ground-state fracton wavefunction~\cite{aletBernevigTNSTEE}. Given the absence of a TQFT description of fracton models, however, it is natural to ask what {\it physical} interpretation may be ascribed to these results.
{Here, we attempt to address this question in the restricted setting of a special class of {\it stabilizer} Hamiltonians built from commuting sets of Pauli operators. For such models, we introduce a concept that we dub the `recoverable information', that is closely related yet distinct from the topological entanglement entropy. For a bipartition $(A,B)$ of a system described by an eigenstate of a stabilizer Hamiltonian, the recoverable information is given by
\begin{align}
\mu= \min_{\mc B}\left[ d_{\partial A}(\mc B) - S_A - S_B \right], \label{mudefintro}
\end{align}
where $S_{A (B)}$ is the entanglement entropy of $A(B)$ and $d_{\partial A}(\mc B)$ is a characterization of the number of stabilizers `cut' by the entanglement surface $\partial A$ for a choice of basis $\mc B$ for the group of all stabilizers. The final computation of $\mu$ entails minimization over all such basis choices. The precise definition of a stabilizer basis and how this minimization is accomplished are elaborated below.

{To understand the relationship between the recoverable information and the topological entanglement entropy, it is useful to recall two distinct interpretations that may be given to the topological entanglement entropy in $d=2$. Recall that the topological entanglement entropy is a negative quantity, indicating that there is {\it less} information lost than
one would na\"ively expect based solely on the area law  for gapped ground states, with no further assumptions. When a TQFT description exists and the topological phase is non-chiral, it is possible to relate the extra information to the geometric structure of the ground state wavefunction --- e.g. for the toric code/$Z_2$ gauge theory, this reflects the closed-loop constraint implied by the Gauss law for the $Z_2$ gauge field. 
Where such a description is absent, clearly we cannot immediately ascribe this additional `topological' information to a physical conservation law. However, by leveraging the special structure of stabilizer Hamiltonians, we show here that we can link the recoverable information to the existence of a set of independent $Z_2$ Gauss law-type constraints that must be satisfied with respect to the entanglement cut, permitting us in some cases to deduce that the phase is described by an underlying generalized `gauge' theory. It is this interpretation that we formalize in this paper.

}
Stabilizer codes were first introduced into quantum information theory in the context of quantum error correction and quantum memory \cite{Gottesman}. A stabilizer code is defined as a subspace of states stabilized, i.e. left invariant, by a set of mutually-commuting Pauli operators ---also known as the \textit{stabilizer set} ---  acting in some tensor product Hilbert space. Clearly, this is also the ground state manifold (`ground space') of the associated stabilizer Hamiltonian given by the negative sum over all members of the stabilizer set. If 
knowledge of the eigenvalues of the stabilizer set does not suffice to specify a complete basis for the Hilbert space, then there exists some complementary set of \textit{logical operators}; these form a Pauli algebra in the stabilizer code, where quantum information can be stored and processed. A logical operator will generally involve a large number of single-qubit Pauli operators, typically as many as characterized by the  \textit{code distance}, that in the ideal scenario for error correction, grows with system size. The resulting non-locality of the information stored via logical operators means that it is robust to errors, making quantum error correction feasible. The special structure of stabilizer codes and the associated mathematical machinery, reviewed in part below, are central in our ability to define and understand recoverable information.

At this point, it is worth remarking that, within the mathematical machinery used to study stabilizer codes, the choice of the stabilizer basis is a key {\it physical} input.  A spatially local basis choice is precisely what makes stabilizer codes useful in a physical setting, since most reasonable perturbations of an actual system will be spatially local. The information stored in logical operators is hence robust against such physical perturbations. Our choice of basis  thus implicitly gives locality a concrete meaning. It will be apparent in what follows that this restriction to a spatially local basis is also what makes the recoverable information a non-zero quantity, since it restricts the  possible allowed measurements  on a system. We will henceforth assume a spatially local basis (suitable defined on a lattice) unless otherwise specified.

The remainder of this paper is organized as follows. First, we review some basic facts about stabilizer codes and the models we consider here, and develop some additional machinery that we require later. We then motivate and define recoverable information and relevant concepts. This leads to the central result of this paper, namely that recoverable information 
counts the number of independent {\it non-local surface stabilizers} (NLSS). These are the generalizations of the non-contractible loop operators generated  in the toric code by the removal of a bounded subsystem. We also show that the NLSSs allow us to identify a set of generalized $Z_2$ Gauss law relations satisfied by the entanglement cut, and obtain a lower bound (tight for some models) for recoverable information. We conclude by calculating the recoverable information of different stabilizer models using three methods: (i) by using the definition; (ii) when tight, the lower bound; and (iii) by counting NLSSs.

\section{Stabilizer Hamiltonians and Fracton Models}\label{sec:models}

We begin this section by reviewing some basic results on stabilizer codes and the particular models we study to set the stage for the discussion, before deriving several new results that will be relevant in the remainder. In certain cases, we will draw upon results from Ref.~\onlinecite{Haah2013poly}, and in those cases we will often only provide details telegraphically or omit them entirely. 

\subsection{Models and Hilbert Space Structure}
Each of the models we study here consists of a set of $N$ qubits (i.e., spin-$1/2$ variables) arranged on the edges or vertices of a {simple} graph, with $q$ qubits per edge or vertex. We denote the graph set (vertices or edges) as $V$, so as to have $N = q|V|$. We use $|\ldots|$ to denote the size of a set. The Hilbert space $\Hi = (\mathbb{C}^2)^{\otimes N}$ is the product space of all the qubits. 
We will consider graphs with either periodic or open boundary conditions.

We define the {\it Pauli group} $P$ to be the set of all Pauli operators (that we denote $\{X,Y,Z\}$) of the qubits modulo any phase of $\pm1, \pm i$. {$P$ can be mapped to a vector space $V[P]$ over the field $\mb F_2$. Addition in $V[P]$ 
corresponds to multiplication of distinct Pauli operators, and scalar multiplication in $V[P]$ (modulo $2$, as appropriate in  $\mb F_2$)  is the power of a Pauli operator; the choice of field should now be obvious since every Pauli operator squares to the identity. For example, in a system of three qubits, $X_1 Y_3 \leftrightarrow \left(\begin{array}{ccc} 1 & 0& 1 \\ 0&0&1\end{array}\right)$, $Z_2 Z_3 \leftrightarrow \left(\begin{array}{ccc} 0 & 0& 0 \\ 0&1&1\end{array}\right)$, and their product $X_1 Z_2 X_3 \leftrightarrow \left(\begin{array}{ccc} 1 & 0& 1 \\ 0&1&0\end{array}\right)$ is given by adding the corresponding $V[P]$-vectors.  Additionally, commutation relations can be encoded by use of a symplectic form, but as we will not use this for what follows, we do not discuss it further. Since the set of single-qubit $X$- and $Z$-type Pauli operators generate all of $P$, we conclude that the dimension of $V[P]$ is $2N$, and that of the Pauli group $P$ itself is $2^{2N}$.} It is useful to define the {\it support} of a Pauli operator $p\in P$ as the set of vertices (edges) upon which it acts nontrivially:
\be
\supp(p) = \{ v \in V: p_v \neq I_{{v}}\},
\ee
where $I_{v} = (I_{2})^{\otimes q}$ is the identity operator on the $q$ qubits on vertex (edge) $v$. Note that we have implicitly defined a map $p_v$ from the set of vertices (edges) to the Pauli group. It is natural to likewise define the support of a set $F\subseteq P$ of Pauli operators, via $\supp(F) = \bigcup_{f\in F} \supp(f)$.

We identify a subset $\mc{S} \subset P$ of the Pauli group as the {\it stabilizer set}. $\mc S$ satisfies the following properties:
\begin{itemize}
\item $ [\mc S,\mc S]=\{0\}$, i.e. it is made of commuting operators
\item $|\mc S| \geq N$, i.e. there are at least as many stabilizers as there are qubits in the system
\item $-I_{\Hi} \not \in \mc S$, i.e. the negative identity is not in $\mc S$, so that all of $\mc S$ can have a positive eigenvalue
 \item $\supp(\mc S) =V$, i.e. every vertex (edge) is acted upon nontrivially by at least one stabilizer in $\mc S$.
\end{itemize}
The Hamiltonian is then taken to be the negative sum of all stabilizers:
\begin{align}\label{eq:stabhamApp}
H= - \sum_{s \in \mc S} s.
\end{align}

For  specificity, we now list the particular examples of stabilizer Hamiltonians that we focus on in this work. Each is defined on a Euclidean lattice of linear size $L$.} While our central focus is on fracton models, for pedagogical reasons, we will also discuss a few more conventional models, both with and without topological order, that serve to illustrate various properties of recoverable information as well as various   considerations necessary for its proper definition. The first of these is  the cluster model. On a $d$-dimensional hypercubic lattice with a single qubit per vertex, this has the Hamiltonian
\begin{align}
H_{\text{clu}_d}=- \sum_v K_v,
\end{align}
where $K_v$ is the product of the single $Z$-type operator at $v$ and all the $X$-type operators surrounding $v$.  (there are $2d$ such $X$ operators on a $d$-dimensional hypercubic lattice). Each stabilizer thus consists of $2d$ $X$-type operators `clustered' around a single $Z$-type operator, hence the name. On a $d$-dimensional hypertorus, the (unique) ground state of this model is not topologically ordered, but is instead either trivial or else in the symmetry-protected topological (SPT) phase~\cite{Raussendorf2005,Doherty2009,Montes2012,Lahtinen2015,Son2011,Cui2013}.

 We also discuss two simple instances of topological order, namely the $d=2$ and $3$ toric codes. These consist of one qubit per edge of the square or cubic lattice, and are described by the Hamiltonian
\begin{align}\label{eq:TCHam}
H_{\text{TC}_d} = - \sum_{v}A_v - \sum_{p} B_p,
\end{align}
where $B_p$ is the product of the four $Z$-type operators surrounding the plaquette $p$, and $A_v$ is the product of the $2d$ $X$-type operators connected to vertex $v$. The entanglement entropy of the $d=3$ toric code was studied in Ref.~\onlinecite{CastelnovoChamon2008}.

The two fracton models we consider are the `X-cube' model~\cite{Vijay2015, Vijay2016} and `Haah's code.'~\cite{Haah2011} The X-cube model is defined on a cubic lattice with a qubit   
on each edge, with Hamiltonian \begin{align}
H_{\text{XC}} = - \sum_{v} \left(A_v^{(xy)} + A_v^{(yz)} + A_v^{(zx)}\right) - \sum_c B_c ,
\label{eq:HXC}
\end{align}
{where $A_v^{(\mu)}$ is the product of the four $Z$-type operators that surround vertex $v$ in one of three mutually orthogonal planes (labeled by their normal direction $\mu$), and $B_c$ is the product of the twelve $X$-type operators surrounding the cube $c$}.  Haah's code is also defined on a cubic lattice, but now with two qubits on every vertex. The $Z$- and $X$-type stabilizer operators now consist of products of $Z$ and $X$  Pauli matrices around a cube with a Hamiltonian of the form
\begin{align}
H_{\text{Haah}} = - \sum_c( G_c^{Z} + G_c^{X}),
\end{align}
{where $G_c^{X(Z)}$ denotes the product of the Pauli operators specified in Fig.~\ref{fig:Haah_code} for cube $c$.}
\begin{figure}[t]
\includegraphics[width=.45\textwidth]{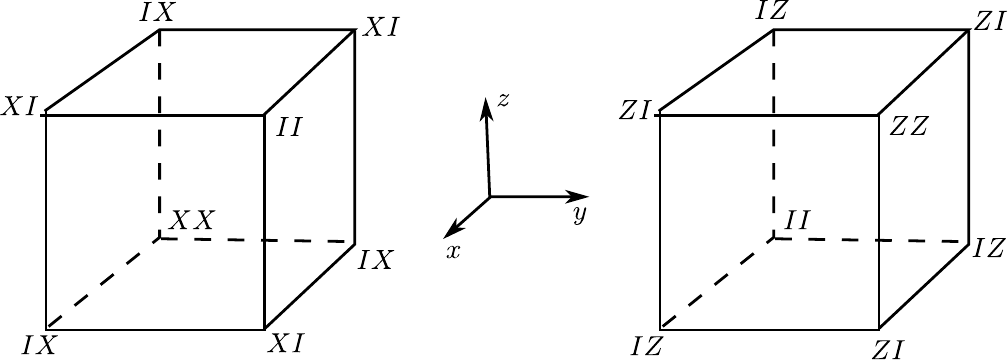}
\caption{$G_c^{X}$ and $G_c^{Z}$ terms in Haah's code. Each site has two spins. $X$ and $Z$ denote the corresponding Pauli operators. $I$ represents the  identity operator.\label{fig:Haah_code} }
\end{figure}
While much of the following discussion will deal with the abstract structure of stabilizer codes, these examples are worth keeping in mind as the particular models of interest, to which we will eventually turn. {Note that we do not discuss many of the striking properties of these fracton models in this work; we refer the interested reader to Refs.~\onlinecite{Chamon2005,Bravyi2011,Haah2011,Yoshida2013,Vijay2015,Vijay2016,Nussinov2009,Nussinov2009(2)} for more details.}

\subsection{Stabilizer Group and Constraint Subspaces}\label{sec:const}

{In this subsection, we review some  properties of stabilizer groups and their constraints. While much of this will be familiar to readers acquainted with stabilizer codes, the next subsections develop technical details relevant to the discussion in Section~\ref{sec:riintro} and so may be useful in the balance of the paper.}

The stabilizer set, $\mc S$, naturally generates the stabilizer group $G= \{\prod_{s \in F} s : F  \in  \mathbb{P}[\mc S]\}$,  where the power set of $\mc S$, denoted ${\mathbb{P}}[{\mc S}]$, is the set of all subsets of $\mc S$. In other words, the stabilizer group is the group generated by taking products of Pauli operators in $\mathcal{S}$. (Note the distinction between the stabilizer set $\mc S$ and the stabilizer group $G$ --- the latter is much larger than the former.)   We also impose an additional requirement of the stabilizer set, that, while not necessary from the perspective of error correcting codes and quantum information, is reasonable from a physical standpoint. Namely, we require that if a bijective mapping of the qubits into themselves induces an automorphism on $G$, then the same  mapping induces an automorphism on $\mc S$ consistent with the automorphism on $G$. In simpler terms, if the group has some symmetry, then we require that the stabilizer set (and hence, the Hamiltonian) has that same symmetry. 
 {Besides being physically reasonable, this has the consequence of imposing certain contraints in the stabilizer set, as we will discuss shortly.} 
 
Now, $G$ is isomorphic to a subspace $V[G] \subseteq V[P]$. 
This follows, again, from the prescription used to define $G$, and is related naturally to the commutative ring formalism~\cite{Haah2013poly} used in Section~\ref{subsec:recoverHaah}. We will denote  the dimension of the vector space $V[G]$ as $d_G$, so that the size of the stabilizer group is $|G| = 2^{d_G}$ (Note that we sometimes use the shorthand $\dim G$; this should be understood as the dimension of the vector space generated by group $G$, rather than the size of the group itself. As noted previously, we  reserve the symbol $|\ldots|$ for the size of a group or a set). 

The definition of $G$ implicitly defines a map  $\phi:\mathbb{P}[\mc S] \rightarrow V[G]$, that takes an element from the power set into $G$. Let us define a basis, ${\mc B} \in \mathbb{P}[\mc S]$, as a minimal generating set for $G$. We emphasize that this basis is a subset of $\mc S$ and does not contain any non-local stabilizers. Now, since $\mc B \in \mathbb{P}[\mc S]$ generates all of $G$ and is minimal, it follows from the definitions above that $\phi(\mb P[{\mc B}])$, the image of the power set of $\mc B$ under $\phi$, is all of the vector space $V[G]$, and thus $\mc B$ maps onto a basis for $V[G]$.
It then follows that $|\mc B| = d_G$. We define a {\it constraint} $C$ on the stabilizer group as a subset of $\mc S$,  $C\in \mathbb{P}[\mc S]$, such that  $\prod_{s\in C} s =I_{\Hi}$. 
  If $d_G <  |\mc S|$, then there must exist $|\mc S| - d_G$ \textit{independent} constraints in $\mc S$.  
{To define independence, we require a vector space structure.}
 Let $W[\mc S]$ be an {\it arbitrary} vector space over $\mathbb{F}_2$ with dimension $|\mc S|$, and hence has $|W| = 2^{|\mc S|}$ distinct elements. For concreteness, we may take $W[\mc S] = \{0,1\}^{|\mc S|}$ where addition is given by bitwise addition modulo 2 ($\oplus$). As this is the same size as $\mathbb{P}[\mc S]$, we can construct a bijection $f: W[\mc S]\rightarrow  \mathbb{P}[{\mc S}]$, with inverse $f^{-1}: \mathbb{P}({\mc S}) \rightarrow W[\mc S] $. In other words,  
 \be
 W[S] \overset{f}{\underset{f^{-1}}{\rightleftharpoons}}\mathbb{P}[{\mc S}] \overset{\phi}{\rightarrow} V[G].
 \ee
Therefore, $\phi \circ f: W[\mc S] \rightarrow V[G]$ is a surjective linear map between vector spaces. Since $f$ is bijective, it follows that the pre-image under $f$ of the kernel of $\phi$ is itself the kernel of $\phi \circ f$, i.e. $\text{ker}(\phi \circ f) = f^{-1}(\text{ker}(\phi))$. Now,  $\text{ker}(\phi)$ contains all sets of stabilizers that map to the identity in $V[G]$, i.e. every element of $\text{ker}(\phi)$ is a constraint. We may now identify the set of constraints as ${\mc C} = \text{ker}(\phi)$. Since $f^{-1}(\mc C) = \text{ker}(\phi \circ f)$ is the kernel of a linear map between vector spaces, it must be a subspace of $W[\mc S]$ and therefore is itself a vector space with its own basis and dimension. So we define a set of constraints to be independent if the preimage under $f$ of this set is independent in $W[\mc S]$. Similarly, we can identify a basis for the constraint set $\mc C$ as the image of a basis for $\text{ker}(\phi \circ f)$ under $f$. Since the orthogonal compliment of $\text{ker}(\phi \circ f)$ consists of every member of $W[\mc S]$ which maps injectively onto $V[G]$, we then have $\dim\, \mc C =   \dim[\text{ker}(\phi \circ f)]  = \dim\, W[\mc S] - \dim\,V[G]  = |\mc S| - d_G$ {as we asserted earlier.}
 We will henceforth refer to $\mc C$  as the `constraint subspace', with the understanding that we refer to this structure. 

\subsection{Logical Pauli Operators} \label{sec:locP}
We next identify a special class of \textit{logical} Pauli operators. If a Pauli operator $p$ anti-commutes with  {any} member of $\mc S$, it creates an excitation above the ground state of the stabilizer Hamiltonian \eqref{eq:stabhamApp}. This implies that $\up{\pj}{\vec{0}} p \up{\pj}{\vec{0}}  =0$, where $\up{\pj}{\vec{0}}$ is a projector onto the ground space. If instead $p$ commutes with all members of $\mc S$,   
then $p$ can be separated as $p= p_\ell p_G$ where $p_G \in V[G]$ and $p_\ell \in V^\perp[G]$, where $V^\perp[G]$ is the complement of $V[G]$ in  $V[P]$. If in this decomposition we have $p_\ell = I_\Hi$, then $p \in V[G]$ and we actually have  $\up{\pj}{\vec{0}} p \up{\pj}{\vec{0}} =\up{\pj}{\vec{0}}$. {This follows since the ground space (the stabilizer code) is stabilized by all of $G$.} 
 So the set of logical Pauli operators, that we denote $ P_\ell$, is the maximal set such that
\begin{itemize}
\item $[\mc S,  P_\ell] = \{0\}$, and
\item $P_\ell \subseteq V^\perp[G]$.
 \end{itemize} 
Given that for any $p$ \textit{not} in $ P _\ell$, $\up{\pj}{\vec{0}} p \up{\pj}{\vec{0}} \propto \up{\pj}{\vec{0}}$, it follows that only logical operators can distinguish between different sectors of the ground state manifold, as any other operator acts trivially when restricted to that manifold. 
This is equivalent to the statement that the quantum information contained in the stabilizer code is difficult to modify without performing non-local operations, and is the origin of their topological robustness.
Furthermore, a maximal set of mutually commuting, independent operators in $ P_\ell$ generates the {\it logical Pauli group}, $G_\ell$. It then follows from the above definitions that $G_\ell$ is isomorphic to a vector space $V[G_\ell]$ over $\mathbb{F}_2$. 
 Since the total number of mutually commuting, independent operators cannot exceed $N$, we may show (see Ref.~\onlinecite{Haah2013poly} for a proof) that $ d_\ell \equiv \dim(G_\ell)  = N-d_G$.

The logical Pauli operators define the {\it code distance}  $D= \min \{d\in \mathbb Z: \supp(p) \subseteq B_v(d) \text{ for some } v \in V \text{ and } p \in P_\ell\}$. Here $B_v(d)$ is a ball about vertex $v$ of radius $d$ using some appropriate distance function on the graph. For example, we could use a Euclidean distance for some embedding of the graph, or the actual graph distance. The choice of distance function must be tuned to the objective. For our purposes, {we wish to relate logical operators to topological information, i.e. non-local correlation, and so we use a Euclidean distance function as it has the most natural link to spatial locality.}  For quantum error correction on the other hand, we might instead restrict the distance function to simply count the smallest number of bit flips needed to change the value of a logical operator (i.e. the size of the support). For all systems considered here and either distance function, we have $D \sim L$, the linear size of the system.   The code distance is the size of the smallest logical Pauli operator, and plays an important role in providing a topological perspective on stabilizer codes, as it gives a sense of which operators are local.  To wit, a subset $A \subseteq V$ is  \textit{bounded} iff there exists a vertex $v$ such that $A \subseteq B_v(D)$, and an operator is bounded iff its support is bounded. A set of operators is {\it local} iff every member is bounded. Since no logical operator is bounded, and these are the only ones that can distinguish between distinct ground states, this shows 
that the ground states are indistinguishable to local operators. This indistinguishability is what makes information encoded in the ground space robust against decoherence, since the sources of decoherence are generically local perturbations that cannot `measure' an eigenstate and are hence unable to decohere it.

{The Pauli group has another important property provided by the \textit{cleaning lemma}~\cite{bravyi2009no,yoshida2010framework,haah2012logical,Haah2013poly,flammia2016limits}. This states that for any bounded subset of the space, $A$, any member of $p_\ell \in P_\ell$ can be `cleaned' out of $A$. That is, there exists an element $g \in G$, such that $gp$ is not supported in $A$. For the example of the toric code on a torus, this corresponds to string-like logical operators which encircle any one direction and can be deformed around any $A$, so long as $A$ itself does not wrap the torus in any direction. The cleaning lemma is actually far more general~\cite{Haah2013poly} than the statement above, but this is sufficient for our purposes.
}
   
\subsection{Topological Constraints and Boundary Conditions} \label{sec:tconst}

With the notion of locality defined as above, we may now apply these ideas to the constraint subspace $\mc C$, and link it with the logical Pauli operators.
 If $ d_G < N$, then from our definition $d_\ell>0$. Furthermore,  since $|\mc S|\geq N$ we have $\dim{\mc C} \geq d_\ell>0$, i.e. the set of constraints is non-trivial.
  
   A set of constraints $\mc T \subseteq \mc C$ is  \textit{topological} if it satisfies the following conditions:
\begin{itemize}
\item $\mc T$ is a subspace of  $\mc C$ (as understood in $W[S]$); and

\item No nonempty member of $\mc T$ remains a constraint if we go from periodic (pbc) to open (obc) boundary conditions.

\end{itemize}

This definition invokes the process of changing boundary conditions from periodic to open that we must make more concrete. We do so  by first discussing the illuminating example of the toric code. With pbc, we have two constraints: the product of all stars and the product of all plaquettes always contain their respective type of single-qubit Pauli operator twice for every edge, and therefore must equal the identity.  When changing boundary conditions, we retain the original Hilbert space, but modify the Hamiltonian. For stabilizer codes, this corresponds implicitly to changing the stabilizer set while preserving all the original qubits. Thus, for the toric code case,  obc requires that we consider a planar graph with two `ragged' edges and two `smooth' edges (see Fig. \ref{fig:bound}); we must then consider carefully all stabilizers affected by this change in boundary conditions. The two ragged edges have complete star operators at each vertex, but incomplete, U-shaped plaquettes; nevertheless, these all commute with each other and with the complete stabilizers in the interior of the graph.
Similarly, the smooth edges have complete plaquettes, but T-shaped `stars' that are all mutually commuting. While the {\it total} number of stabilizers is unchanged, both constraints are lost, since some of the Pauli operators on the edges only appear once when taking the product of all the star or plaquette stabilizers. Simultaneously, we observe that the number of stabilizers is equal to the number of qubits so that there can be no logical operators. Instead, the product of all stars (that equals the identity with pbc) is now an L-shaped string-like  operator on the two perpendicular ragged edges; this is the descendant of what would be a {logical} operator with pbc but is included in the stabilizer group with obc. A similar statement holds for plaquettes. Thus, the change of boundary conditions provides a natural link between constraints and logical operators: every time a change in boundary conditions removes a constraint, a logical operator is transformed into a string operator included in the {\it modified} stabilizer group.

We may formalize this intuition by defining a {\it change of boundary conditions}  to be a map from the original set of stabilizers $\mc S$ to an alternative stabilizer set $\mc S^\prime$. $\mc S^\prime$ satisfies all the conditions of a stabilizer set, but generates a group  $G^\prime$ that is not equal to $G$. We also require that $|\mc S \cap \mc S^\prime| \sim \mc O (L^{d})$, i.e. their intersection is extensive (well-defined on a lattice).  Finally, for every member $s^\prime \in \mc S^\prime$ there exists a unique member $s \in \mc S$ such that $\supp_\mu(s^\prime) \subseteq \supp_\mu(s)$, where we use the subscript $\mu$ to emphasize that we consider the support of the same Pauli operator in each case. Thus,  every member of 
$\mc S$ has a corresponding member of 
$\mc S'$ which is either the same or has some ``deleted'' support. As a corollary to the last condition $|\mc S| = |\mc S^\prime|$. If there exists a change of boundary conditions  $\mc S\rightarrow\mc S^\prime$ such that a subspace of $\mc C$  (the constraints in  $\mc S$) is lost, then the change has removed topological constraints. The maximal set of constraints lost by a {set of equivalent changes of boundary conditions (where equivalency is related by symmetry where it exists)} is then the set $\mc T$ of topological constraints. Clearly $ \dim(\mc  T)\leq d_\ell$ since every time the constraint subspace dimension is lowered by one, $d_G$  increases by one until it is equal to $N$, at which point one has exactly enough independent stabilizers to specify a complete eigenbasis for $\Hi$. Since a change in boundary conditions maintains the same number of stabilizers, all other constraints beyond a number $d_\ell$ must necessarily remain unchanged, leading to the bound. Therefore, we conclude that the subspace of topological constraints is isomorphic to a subset of  the logical Pauli group $\mc T  \cong G_{\mc T}\subseteq G_\ell$. 

It is worth noting that {\it typically} any topological constraint will `wrap' all the way around the system (defined with periodic boundary conditions) and thus will have non-local support in the thermodynamic limit. This justifies an implicit assumption that we make below: namely that non-topological constraints are typically (though possibly not always) local constraints. Note that in all the examples studied in Section~\ref{sec:compute}, this can be verified explicitly. {In the remainder, we will assume that topological constraints are in one-to-one correspondence with logical operators.}

\subsection{Entanglement Entropy of Stabilizer Eigenstates} \label{sec:ee}

We now calculate the entanglement entropy of any eigenstate of the stabilizer Hamiltonian {(i.e., {\it any} simultaneous eigenstate of all members of $G$ and $G_\ell$)}  for a bipartition $(A,B)$ of the space $\Hi$. This derivation reviews and extends the results from Refs.~\onlinecite{zou-haah,hamma2005bipartite,fattal2004entanglement}; more details and explicit computations using this approach may be found in our previous work~\cite{companion1}. Let $A$ be bounded as defined in Section \ref{sec:locP}. Also let our basis for $G$ be $\mc B$ as before and $\mc B_\ell$ be a basis for $G_\ell$. With the cleaning lemma, we can assume without loss of generality that the support of $\mc B_\ell$ is contained entirely in $B$.

{
As before, let $W[\mc B] = \{0,1\}^{d_G}$ so that elements $g\in G$ may then be labeled by a binary vector $\vec{n} \in W[\mc B] \simeq V[G]$, via
\be\label{eq:gofndef}
g(\vec{n}) =  s_1^{n_1} s_2^{n_2}\ldots s_{d_G}^{n_{d_G}}
\ee
where $s_i \in \mc B$ and $n_i$ is the $i^{th}$ bit of the string $\vec{n}$. Note the ordering of Eq. \eqref{eq:gofndef} defines the isomorphism for $ W[\mc B] \rightleftharpoons V[G]$. As this is isomorphic, we see that $(W[\mc B], \oplus)$ is a faithful representation of $G$. Furthermore, one can see that the irreps of $G$ can be labeled by another string $\vec{k} \in W_{\text{irrep}}[\mc B]= W[\mc B]$ and given by $\{(-1)^{\vec{k} \cdot \vec{n}}\}_{\vec{n}}$, where $\vec{k}\cdot \vec{n}$ is the binary dot product. This should be obvious as any Abelian group has 1D irreps and one can show that 
\begin{align}
(-1)^{\vec{k} \cdot (\vec{m} \oplus \vec{n})} = (-1)^{\vec{k} \cdot \vec{m}}(-1)^{ \vec{k} \cdot \vec{n}}.
\end{align}
Thus they preserve the multiplication rules for $W[\mc B]$. It then follows from standard orthogonality relations that projection operators can be formed as
\be \label{eq:proj}
\mathcal{P}^{(\vec{k})} = \frac{1}{|G|} \sum_{\vec{n}} (-1)^{\vec{k}\cdot\vec{n}} g(\vec{n}).
\ee
As the stabilizer Hamiltonian is just a sum over group elements, it is easy to show that these operators project onto the various eigenvalue subspaces as labeled by the quantum number, $\vec{k}$. Clearly, this labels the stabilizer excitations and so $\vec{k} =0$ is the ground space i.e. the one with no excitations. To find the dimension of these subspaces, we can take the trace and use the fact that $\tr(g(\vec{n})) = N \delta_{\vec{n},0}$, to find that for any $\vec{k}$, the \textit{topological} degeneracy is $N-d_G = d_\ell$, as expected. (Note that for a translationally-invariant code there are additional degeneracies among the various values of $\vec{k}$, hence our use of the term {\it topological} degeneracy to distinguish that linked to the logical Pauli operators). To obtain pure state projectors and hence the density matrix, we can construct analogous projection operators using $\mc B_\ell$, %
\begin{align}
\mathcal{P}^{(\vec{a})}_\ell =  \frac{1}{|G_\ell|} \sum_{\vec{z}} (-1)^{\vec{a}\cdot\vec{z}} \mf g(\vec{z}),
\end{align}
where $\vec{a}, \vec{z} \in W[\mc B_\ell]$ and $ \mf g(\vec{z}) \in G_ \ell$ as given by a basis expansion analogous to Eq. \eqref{eq:gofndef}. This gives the pure state density matrix 
\begin{align}
\outerp{\vec{k}, \vec{a}}{\vec{k}, \vec{a}} =\mathcal{P}^{(\vec{k})} \mathcal{P}^{(\vec{a})}_\ell
\end{align}
With an explicit form of the density matrix, we take the partial trace to get the reduced density matrix of $A$ as
\begin{align}\label{eq:rhoastep1}
\rho_A =& \Tr_B \outerp{\vec{k}, \vec{a}}{\vec{k}, \vec{a}} = \nonumber \\
=&\frac{1}{|G| |G_\ell|} \sum_{\vec{n}, \vec{z}} (-1)^{\vec{k}\cdot\vec{n}} (-1)^{\vec{a}\cdot\vec{z}} \Tr_B \left(g(\vec{n}) \mf g(\vec{z})\right). 
\end{align}
Any element $ g(\vec{n}) \mf g(\vec{z})$ that is {\it not} equal to the identity on $B$ contains at least one $X$ or $Z$ operator acting in $B$ and thus has a zero partial trace. The only non-zero contributions to the sum in Eq. \eqref{eq:rhoastep1} is from operators supported only on $A$, i.e. equal to the identity on $B$.  For the logical group, only $\vec{z} =0$ survives as any other value of $\vec{z}$ is not the identity in $B$. However, there is an entire subgroup, $G_A \subseteq G$, which survives the trace. Such a subgroup is defined as all elements of $G$ which only have support in $A$. As a subgroup, it has an analogous basis labeled by $\vec{n}_A$ and excitations labeled by $\vec{k}_A$, where we assume that $\vec{k_A} \cdot \vec{n}_A = \vec{k} \cdot \vec{n} \mod 2$ for the same element as labeled by $\vec{n}_A$ in $G_A$ and $\vec{n}$ in $G$. Since the identity on $B$ has trace $\Tr I_B = 2^{|B|}$ we have that
\bea
\rho_A &=& \frac{2^{|B|}}{|G||G_\ell|} \sum_{n_A}  (-1)^{\vec{k}_A \cdot \vec{n}_A} g_A(\vec{n}_A)\nonumber\\ 
&=&\frac{2^{|B|}|G_A|}{|G||G_\ell|}\up{\pj}{\vec{k}_A}_A.\label{eq:rhoastep2}
\eea
 Since Eq. \eqref{eq:rhoastep2} expresses $\rho_A$ as a projector, its entanglement entropy follows straightforwardly:
\be\label{eq:SAcount}
S_A = (|B| - d_G -d_\ell) -d_{G_A} = |A| - d_{G_A} ,
\ee
where $d_{G_A} = \dim G_A$  and we have used the fact that $N = |A| + |B| = d_G + d_\ell$. 

This same calculation can be carried out from the perspective of $B$, and we would similarly find that there exists a subgroup, $G_B \subseteq G$, consisting of all elements of $G$ which only have support in $B$. The only difference in the argument is that all of $G_\ell$ survives the partial trace over $A$, so that the reduced density matrix for $B$ is now a projection defined by the group $G_B \times G_\ell$. Here, the $\times$ product of two groups generates all products between elements of the two groups. Note that taking such a product of trivially intersecting Abelian groups generates the $\mb F_2$ vector  space given by the direct sum of each individually, i.e. $V[G_B \times G_\ell] = V[G_B] \oplus V[G_ \ell]$. Defining $d_{G_B} = \dim G_B$, one has that the entanglement entropy for $B$ is
\begin{align}\label{eq:SBcount}
S_B = |B| - d_{G_B} -d_\ell.
\end{align}
Even though $S_A = S_B$ as the original state was pure, we will use both forms in the discussion of recoverable information.
}

\section{Recoverable information in stabilizer codes}\label{sec:riintro}

We now introduce the concept of `recoverable information.'   In a subsequent section we will show that for stabilizer codes, this quantity is intimately tied to {topological properties of the stabilizer set, and hence to those of the Hamiltonian naturally generated by it.}  
 Furthermore, it can be calculated with information obtained entirely from system $A$, and in this sense, is locally measurable.

 To motivate the idea of recoverable information, consider the following situation. A binary password  
is encoded into the eigenstate of a stabilizer code of linear size $L$ (the { $\sim L^d$} bits of the password being simply the eigenvalues of the stabilizer and logical operators). Then 
an $R^d$ sub-region $A$ of the eigenstate is removed and given to Alice, with the remaining piece $B=\bar{A}$ given to Bob. Alice and Bob are only allowed local operations and measurements, while afterwards, they are allowed to combine all measurement results via classical communication.  
How many bits of the binary password can they extract? 

 Since the stabilizers form a commuting set and the wavefunction is an eigenstate of the stabilizers, Alice can faithfully and without disturbing the state measure all the bits of the password encoded in stabilizers contained entirely in her subsystem. Bob can do the same with its complement. However, because Alice's density matrix is proportional to a projection, she is in a completely mixed state for all the degrees of freedom which lie on the boundary of A. That is, if she tries to measure any of the partial stabilizers that were cut when separating $A$ and $B$, she could measure $\pm1$ with equal probability. Na\"{i}vely, Alice and Bob might assume they have lost all the information about cut stabilizers, but the {\it actual} amount of information lost is given by the entropy for each of their states. There is thus a difference between the na\"ive and actual information loss, viz.
 \begin{align}
\mu(\mc B)= 
d_{\partial A}(\mc B) - S_A - S_B, \label{mudef}
\end{align}
where $d_{\partial A}$ is the number of basis stabilizers cut by the entanglement surface.   Eq. \eqref{mudef} emphasizes that value of $d_{\partial A}$ can depend on the choice of stabilizer basis, and a `bad basis' in which the stabilizers are less spatially local will lead to  more cut stabilizers and hence a higher value of $\mu$.  So, in order to unambiguously define the  recoverable information we  minimize \eqref{mudef}  over all possible choices of stabilizer basis (effectively assuming obc as discussed below). This removes the basis-dependence, leading to the definition \eqref{mudefintro} quoted in the introduction. This definition implicitly involves the `most spatially local' set of stabilizers (analogous to the `clipped stabilizers' of Ref.~\onlinecite{Nahum}). A second caveat also implicit in this definition is that the set of cut stabilizers should not include any {\it logical} operators. This is trivially true for a system with open boundary conditions. For a system with periodic boundary conditions, this may be ensured by working in a basis where the logical operators are contained entirely in $B =\bar{A}$ as was assumed in Section \ref{sec:ee}.

This discussion motivates a complementary perspective on recoverable information. Imagine  we cut the system as described, but also {\it forbid} Alice and Bob from communicating by any means. They then perform all their measurements of any combination of stabilizers accessible to them, and use this to estimate how much information they have lost about the system. Then, classical two-way communication between Alice and Bob is reestablished, allowing them to compare their measurements; it is possible that in the process, they can get some additional information that was missing when they were forbidden from communicating. It is this additional information that is the `recoverable' information.

How may we compute the recoverable information of a bipartition? The 
most na\"ive procedure is to count cut stabilizers (suitably defined) to determine $d_{\partial A}$, and then compute the entanglement entropy using the methods of Section~\ref{sec:ee}.  An alternative approach is to count topological constraints. This provides a lower bound on $\mu$ that is tight for some of the models studied here. Finally and most powerfully, we will show in this section that $\mu$ may be determined directly by counting the number of independent NLSSs. 
Each NLSS is typically a product of stabilizers in $A$ and on its boundary, which has support only in $B$. We will demonstrate the use of each of these methods in the examples of Section~\ref{sec:compute}.

\subsection{Warm-up: the $d=1$ cluster model and the $d=2$ toric code} 

Before embarking on an abstract analysis of recoverable information, we discuss two simple illustrative examples, the $d=1$ cluster model and the $d=2$ toric code. The $d=1$ cluster model with periodic boundary conditions is not topologically ordered (i.e., it is short range entangled and has no ground space degeneracy), and so we should expect that its recoverable information is zero. Note that it also has no constraints, topological or otherwise. We take subsystem $A$ to be a contiguous set of $R$ qubits. Every stabilizer in $A$ except for the two vertices at either end are completely in $A$ and so applying Eq. \eqref{eq:SAcount} we find the entropy is $S_A =2$. As for cut stabilizers, both edge vertices in $A$ and edge vertices in $B$ represent cut stabilizers, in which case the recoverable information is
\begin{align}
\mu_{\text{clu}_1} = 4- 2(2) = 0.
\end{align}
So in this case, Alice and Bob gain no additional information when allowed to communicate classically. We will discuss cluster models in $d>1$ in Sec. \ref{sec:computeclu}.

We now turn to the more interesting case of the $d=2$ toric code. We take subsystem $A$ to be the square $R \times R$ cut (in edges) so that Alice does not have `ragged' edges, but Bob does. This is depicted in Fig.~\ref{fig:bound}. Then, Alice has parts of cut plaquette and vertex operators. {Consider only the cut vertex operators for Alice. Since they still commute with all complete stabilizers in $A$,} they still represent good qubits, the measurement of which will not disturb the other stabilizer qubits. Likewise, the cut plaquette operators still commute with all of Bob's stabilizers and thus are good qubits for him. 
There are $4R$ such good boundary qubits for Alice and $4R$ good boundary qubits for Bob, all of which are apparently in a mixed state ($\pm1$ with equal probability). However, the actual amount of lost information for either Alice or Bob is given by the entanglement entropy, $S_A =4R -1$. We thus conclude that the recoverable information is 
\begin{align}
\mu_{\text{TC}_2} =2 \times 4R - 2(4R-1)) = 2 ,
\end{align}
which is twice the topological entanglement of the toric code.  The existence of this recoverable information comes from the fact that there is a constraint between Alice's stabilizers and boundary qubits, namely that
\begin{figure}[t]
\centering
\centering
\includegraphics[scale=.2]{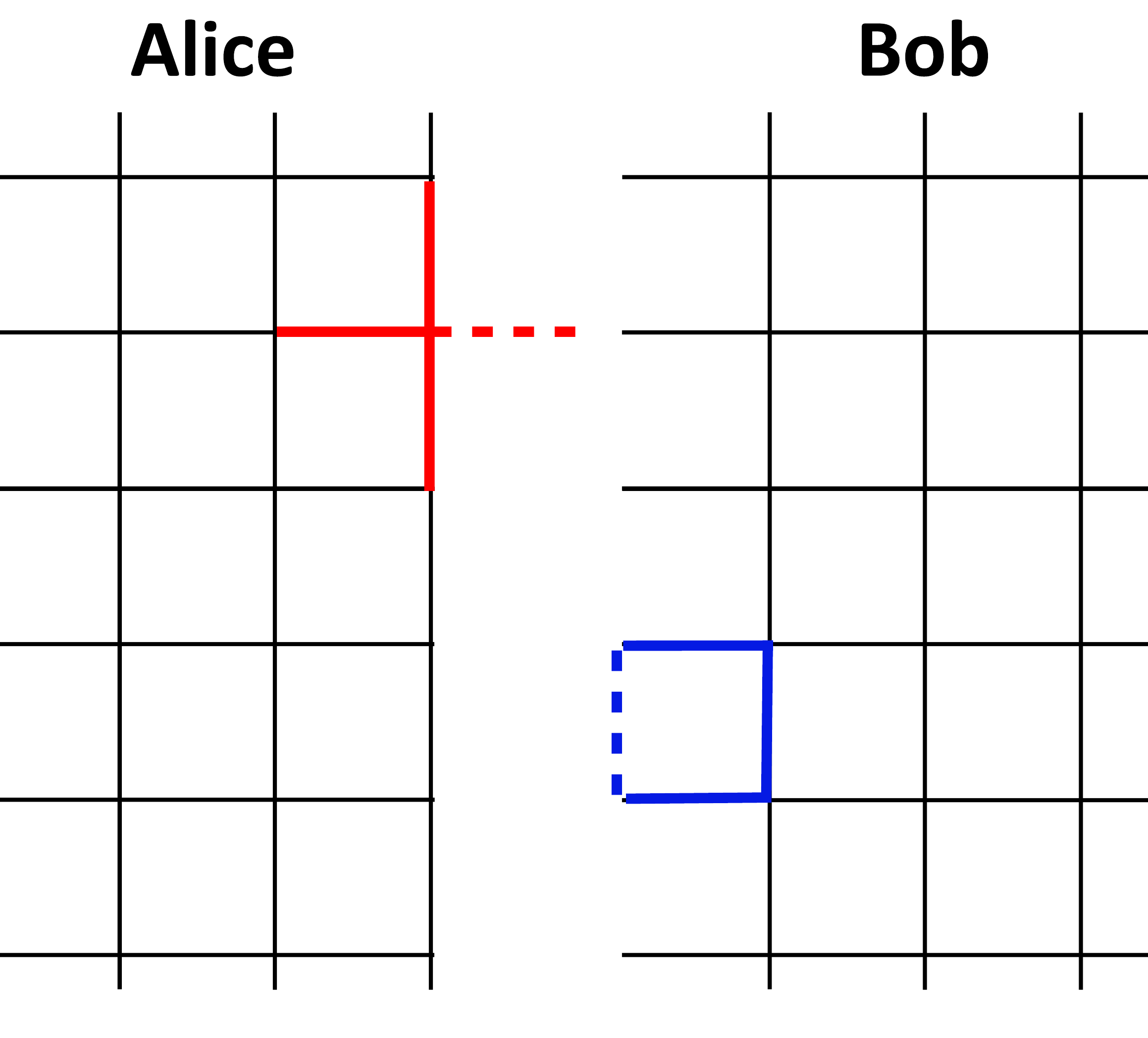}
\caption{Depiction of the boundary of Alice and Bob's system. The red (blue) edges represent a good cut stabilizer for Alice (Bob). The dotted edges are not actually present but represent the completion of the cut stabilizer.}\label{fig:bound}
\end{figure}%

\begin{align}\label{eq:toricgl}
\prod_{v \in A} (A_v)_A =I_A,
\end{align}
--- in other words, the product of all the vertex stabilizers, both cut and whole, is the identity in $A$. This implies that Alice has one additional bit of information in her system. To understand the nature of the information recoverable by Alice and Bob, and how they can recover it in practice, note that  there are global constraints on the stabilizers, namely 
\begin{align}
\prod_p  B_p = I_{AB} \text{ and } \prod_v A_v= I _{AB}.
\end{align}
Alice can measure her vertex operators (all via local operations) and communicate the results (purely classical information) to Bob. Bob  now measures all his vertex operators and together with the constraint, 
%
\begin{align}
\prod_{v\in \partial A} A_v = \prod_{v\in \text{bulk}(A)}A_v \prod_{v\in \text{bulk}(B)}A_v
\end{align}
In this way, Alice and Bob have gained one additional piece of information about the system than they would na\"ively have expected, since they now know the {\it parity} of the boundary stabilizers. Likewise, they can do the same for the plaquettes, so that the actual information gained is two bits. We may also attribute this result to the presence of two NLSSs. Each consists of a string loop (one of each type, $X$ and $Z$) in $B$ which are non-contractible due to the hole created by removing $A$. These strings are created by the product of all stabilizers of a given type in $A$ and on the boundary.  We also note that we have cut two independent topological constraints.

To relate this back to topological considerations, we mentioned in the introduction that the topological entanglement entropy of the toric code can be related to a geometric constraint on the wavefunction. In the loop condensation picture, the net parity of edges perpendicular to the boundary as measured in the $X$ basis must be $1$, i.e. every loop in the condensate is cut twice by the boundary. This is shown to be the source of the one bit of topological entanglement entropy. However as we have shown, this bit is not special to a ground state: all {simultaneous eigenstates of the stabilizers (and members of the logical group)} have the same entanglement entropy and thus the same topological entanglement entropy. This is because the constraint for the ground state is generalized to a $Z_2$ Gauss law for excited states. Namely, the parity of the boundary is equal to the number of topological charges within the bounded region modulo $2$. An equivalent statement can be made for the magnetic sector. To see that these two Gauss's laws correspond to the two bits of recoverable information, we rewrite Eq. \eqref{eq:toricgl} as
\begin{align}\label{eq:toricgl2}
\prod_{v \in \text{bulk}(A)} A_v = \prod_{v \in \partial A} (A_v)_A.
\end{align}
This is exactly  the electric $Z_2$ Gauss's law. An analogous statement can be made for the magnetic sector using plaquette stabilizers. We show that this is not coincidental; recoverable information for other stabilizer models  also matches the number of independent $Z_2$ Gauss law-type relations satisfied by the entanglement cut. {We return to the $d=3$ toric code in Sec. \ref{RIforTC3D}.
}

\subsection{Recoverable Information and Topological Properties of Stabilizer Codes \label{sec:relation_recover}}

With these preliminaries, we may now clarify various facts about recoverable information. 
{In Section~\ref{sec:riintro}, (see \eqref{mudef} and the subsequent discussion)
we defined the recoverable information via a process of minimization over choices of stabilizer basis, leading to \eqref{mudefintro}, reproduced here for convenience: 
\begin{align}\label{eq:fullmudef}
\mu= \min_{\mc B} [d_{\partial A}(\mc B) - S_A - S_B].
\end{align}}
We now discuss how to compute the recoverable information via three different methods.
 
 \subsubsection{Computing recoverable information by counting cut stabilizers}
 Using \eqref{eq:SAcount} and \eqref{eq:SBcount} in  \eqref{eq:fullmudef}, and 
 $|A| +|B| = N = d_G + d_\ell$ again, we have
{\begin{align}
\mu=\min_{\mc B} [ (d_{G_B} + d_{G_A} + d_{\partial A}(\mc B)) - d_G]. \label{muinter}
\end{align}}

 To go further and  {compute $\mu$ directly by counting stabilizers} we need to clarify the meaning of $d_{\partial A}(\mc B)$, which will also motivate our minimization procedure for \eqref{muinter}. 
{Recall that we chose the basis $\mc B$  to be a subset of $\mc S$; this is always possible as we generated $G$ from $\mc S$.} Given any basis $\mc B$, we can clearly partition it as $\mc B = \mc B_{A}\cup \mc B_B\cup   \mc B_\partial $  where $\mc B_A=\{g \in \mc B: \supp(g) \subseteq A\}$, $\mc B_B =\{g\in \mc B:\supp(g) \subseteq B\}$, and $\mc B_ \partial =\{ g\in \mc B: \supp(g) \cap A \neq \emptyset \text{ and } \supp(g) \cap B \neq \emptyset\}$, respectively corresponding to basis elements entirely in $A$, entirely in $B$, and supported in both. Note that it is not necessarily true that $G_A \cong \text{Span}({\mc B}_{A})$ and $G_B\cong \text{Span}({\mc B}_{B})$: the subgroup of the stabilizer group obtained by  restricting its support to some subregion need not coincide with the subgroup generated by the restriction of the basis to that subregion.
Below, we will frequently and without explicit comment use subscripts $A, B$ and $\partial$ to denote the partitioning of a set into elements only supported in $A$, only supported in $B$, and supported in both $A$ and $B$.
It follows that 
 \be d_G = |\mc B_A| +|\mc B_B| + |\mc B_\partial|. \label{eq:dGbasisdef}
 \ee
 $\mc B_\partial$ in turn generates the {\it cut stabilizer group} {$G_\partial= \spn(\mc B_\partial)$} with dimensions $ |\mc B_\partial|$,  motivating  us to define $d_{\partial A} = |\mc B_\partial|$. {Now, we may easily show that  $G_A \times G_B \times G_\partial =G$; note that {$\text{dim } G$ is not necessarily equal to $\text{dim } G_A +\text{dim } G_B + \text{dim } G_\partial$}. Then, from \eqref{muinter}, we see that the only way for the recoverable information to be nonzero is if there is some nontrivial pairwise intersection between  $G_A$, $G_B$ and $G_\partial$ (for, if they were disjoint then $d_G = d_{G_B} + d_{G_A} + d_{\partial A}$). 
 Note that $G_A$ and $G_B$ correspond to disjoint subspaces, since any product of two elements of $ G_A$ is also in $G_A$, and likewise for $G_B$, and the support of $G_A$ is disjoint from that of $G_B$.   
  
We note that the definition of $G_\partial$} is basis-dependent, but this is expected given the motivation for recoverable information:  the encoding of a password in stabilizer quantum numbers is inherently basis-dependent.  
{However, we can restrict the basis choice by exploiting constraints and the definition of recoverable information from Section~\ref{sec:riintro}. By doing so, we automatically perform the minimization procedure that makes the recoverable information unambiguous. 
Our approach, elaborated below, will also resolve subtleties in quantifying the quantum information encoded in sets of stabilizers that are required to satisfy local constraints.}   

From the definition of $d_{\partial A} = |\mc B_\partial|$ and \eqref{eq:dGbasisdef} we see that \eqref{muinter} reduces to 
{\begin{align}
\mu=& \min_{\mc B} [(d_{G_A} - |\mc B_A|) + (d_{G_B} - |\mc B_B|)]  \nonumber \\
=& \min_{\mc B} [\dim\left(G_A / \spn(\mc B_A) \times G_B / \spn(\mc B_B)\right)]. \label{eq:mufin}
\end{align}}
{The second equality follows from the fact that $G_A / \spn(\mc B_A)$ and  $G_B / \spn(\mc B_B)$ are disjoint subgroups of $G$ with dimensions $d_{G_A} - |\mc B_A|$ and $d_{G_B} - |\mc B_B|$, respectively. Note \eqref{eq:mufin} proves that  $\mu$ is the {(minimum)} dimension of a group and is hence non-negative, $\mu \geq 0$.

We have yet to fully fix the basis that minimizes $\mu$. To do so, and fully remove ambiguity in the basis choice,  we recall an important fact: namely, that the difference between the full stabilizer set $\mc S$ and  the basis $\mc B\subseteq \mc S$   is linked to the existence of constraints. So we can think of creating a basis as a process of removing stabilizers: {we need to remove exactly one stabilizer for each independent constraint.} 

 To proceed, we define the following: a Pauli operator, $p$, is said to be {\it cut} by a bipartition $(A,B)$ iff $\supp(p) \cap A \neq \emptyset$ and $\supp(p) \cap B \neq \emptyset$; a set of Pauli operators is cut iff at least one member is cut. Constraints that are not cut by definition contain only elements entirely in $A$ or $B$, and so which stabilizer is removed to form the basis is inconsequential. However, when a constraint is cut, we have the choice to remove either a non-cut stabilizer or a cut stabilizer, leading to some remaining ambiguity in the basis choice. We propose the following rules to remove this ambiguity and implement the minimization which defines recoverable information in Eq. \eqref{eq:fullmudef}.
The rules are simple: when a constraint is cut, we remove a non-cut (cut) stabilizer if the constraint is topological (non-topological). (Note that typically, a non-cut stabilizer removed due to a cut topological constraint will be contained entirely in $B$.)} 
 In practice, {our prescription} means that we {simply compute $d_{\partial A}$ by counting cut stabilizers, with the caveat that we} reduce the number of cut stabilizers by one for every cut non-topological constraint. Note that per this algorithm, when subtracting cut constraints, we ignore {any} topological constraints (effectively assuming obc),
 {since they involve removing only non-cut stabilizers from  $B$}. Since decreasing $d_{\partial A} = |\mc B_\partial|$ decreases $\mu$, this {procedure yields} the minimal value of $\mu$, modulo changes to the boundary conditions.

 We may justify this method of generating the basis as follows.   When the constraint is topological, {the relation between topological constraints and logical Pauli operators tells us that} there must exist a logical operator which contains a piece of the code.  
 The logical bit implied by the topological constraint is  always recoverable via the cleaning lemma, so {always contributes to $\mu$.} {(Note that this is consonant with the fact that a topological constraint will never reduce the count of cut stabilizers in our algorithm above)}. {
 Put differently, we are effectively minimizing the recoverable information  over \textit{all possible {\rm local} bases used to encode the information}, thereby retaining only the non-local information that is recoverable for any local basis choice. The basis connection lemma, proved in Appendix~\ref{appendixproofs}, ensures that this procedure will not get ``stuck'' in local minima and instead will converge on a globally minimizing basis that yields $\mu$.

\subsubsection{A bound on the recoverable information}
As a corollary, we also derive a useful {\bf  bound on $\mu$}, providing another route to characterizing recoverable information. Define the set of {\it cut topological constraints},  $\cut (\T) = \{C\in \T: C \text{ is cut by } (A,B)\}$. Then
\begin{align}
\mu \geq \dim(\cut(\T)).
\end{align}
This is obvious as we specifically did not remove cut stabilizers when the constraints involved were topological. 

\subsubsection{Computing recoverable information from non-local surface stabilizers}
Finally, we can relate recoverable information to a special subgroup of $G$ defined for a bipartition $(A,B)$ and a basis $\mc B$, that we term the {\it non-local surface stabilizer group}, denoted $G_{\text{NLSS}}$. This group captures all of the non-trivial overlap between the groups $G_A, G_B$ and $G_\partial$ and is defined as
\begin{align}
G_{\text{NLSS}}= G_\partial \cap (G_A \times G_B).
\end{align}
It follows from this definition that any element $g_\partial \in G_{\text{NLSS}}$ may be generated in one of two ways:
(i) as a product over some set of cut basis members, $F_\partial \subseteq \mc B_\partial$ and (ii) as a product of some elements $g_A \in G_A$ and $g_B \in G_B$, i.e. $g_\partial = \prod_{s\in F_\partial}s= g_A g_B$. 
We show (via relation  \eqref{eq:NLSSfactor}) that given $G_{\text{NLSS}}$, we can always choose a basis $\mc{B}$ such that}  either $g_A$ or $g_B$, but not both, can be written as a product over members of $\mc B_A (\mc B_B)$. Thus we have a product of stabilizers in $A$ ($B$) and on the boundary which only has support in $B(A)$.

As in the $d=2$ toric code, the elements of $G_{\text{NLSS}}$ are topological in nature, and are directly linked to the Gauss's law-type constraints satisfied by the entanglement cut.
To see this, let us suppose all of $G_A$ is generated by stabilizers which are confined to $A$ (as is usually the case when $A$ is bounded and simply-connected), and assume that we choose a basis $\mc B$ that contains all such stabilizers. Now consider any element $ g_\partial \in G_{\text{NLSS}}$; from our discussion above, there exist elements $g_A\in G_A$ and $g_B \in G_B$ so that $g_\partial = g_A g_B$. By restricting both sides of this expression to its support in $A$ (denoted $(\cdot )_A$) 
we find $(g_\partial)_A = (g_A g_B)_A =(g_A)_A (g_B)_A = g_A$, as $g_B$ is supported only in $B$ by definition.  By our assumptions about $G_A$, $g_A$ is generated by a basis set $F_A \subseteq \mc B_A$, i.e. $g_A = \prod_{s\in F_A} s$, leading to the operator relation
\begin{align} \label{eq:gengl}
\prod_{s\in F_A} s = (g_\partial)_A.
\end{align}
Since each member of $F_A$ is a local stabilizer that squares to the identity, it can be thought of as a $Z_2$ charge. In this sense, \eqref{eq:gengl} expresses the fact that  the number of charges modulo $2$ within some subset of the bulk of a compact subregion $A$ is equal to a measurement on the boundary of $A$ --- a relation akin to Gauss's law. Thus, any NLSS can be mapped to a Gauss's law-type relation (or combination of such relations) satisfied by the entanglement cut. Though we assumed $g_A$ was formed by stabilizers in $A$, %
we can see that if $g_B$ is formed by a set $F_B \subseteq \mc B_B$ instead, an analogous statement holds when the equation is restricted to its support in $B$.

Note that $G_{\text{NLSS}}$ is dependent on the basis (which we say {\it generates} $G_{\text{NLSS}}$); we prove in Appendix \ref{appendixproofs} that 
\begin{align} \label{eq:NLSSeq}
G_{\text{NLSS}} \simeq G_A/ \spn(\mc B_A) \times G_B / \spn(\mc B_B).
\end{align}
Comparing this to Eq. \eqref{eq:mufin}, we see non-local surface stabilizer groups can be related to recoverable information by the following equivalent statements:
\begin{enumerate}[label= (S\arabic*)]
\setcounter{enumi}{0}

\item $\dim( G_{\text{NLSS}}) = \mu$ (see the minimization definition in \eqref{mudefintro} along with Eqs. \eqref{eq:mufin} and \eqref{eq:NLSSeq}), \label{s1}

\item For any element $g_\partial \in G_{\text{NLSS}}$ with its associated $g_A \in G_A$ and $g_B \in G_B$ for which $g_\partial = g_A g_B$, if there exist sets of {stabilizers}, $F_A \subseteq \mc S_A$ and $F_B \subseteq \mc S_B$  
 whose product is $g_A, g_B$ respectively, then $F_A \cup F_B \cup F_\partial \in \mc T$, where $F_\partial \subseteq \mc B _\partial$ is the basis set which generates $g_\partial$ and \label{s2}

\item For all $g_\partial  \neq I \in G_{\text{NLSS}}$ such that $F_\partial \subseteq \mc B _\partial$ is the basis set which generates $g_\partial$, if there exists a {non-topological} constraint $C \in \mc C -  \mc T$ such that $F_\partial \subseteq C$, then $C_\partial - F_\partial \neq \emptyset$. \label{s3}

\end{enumerate} 
Recall $\mc C$ is the set of all constraints and $\mc T$ is the set of topological constraints as discussed in Sections \ref{sec:const} and \ref{sec:tconst}. See Appendix \ref{appendixproofs} for proof of the equivalence of these statements. \ref{s2} and \ref{s3} can be useful for computing $\mu$ as they make no reference to a generating basis for $G_{\text{NLSS}}$ (outside of $\mc B_\partial$, but generally, this can be replaced by ${\mc S}_\partial$, the set of stabilizers at the boundary). For example, suppose we find some independent set of stabilizer group elements which look like NLSSs (i.e. satisfy a relation of the form $g_\partial = g_A g_B$). However, there is no immediate way by which we can check that these form some $G_{\text{NLSS}}$ outside of finding a basis which generates it, let alone  assess whether such a basis is minimal.
 However, if our set of NLSSs satisfy statements \ref{s2} or \ref{s3}, we know that its dimension is a lower bound for $\mu$. Then we can continue to add independent NLSSs to our candidate set until \ref{s2} and \ref{s3} are violated, at which point we must have had a $G_{\text{NLSS}}$ which is generated by a minimal basis. This is useful in cases where calculating $\mu$ via the definition \eqref{eq:fullmudef} is considerably more difficult than finding possible NLSSs.

Given this result, it is useful to consider what the minimization procedure implies for the emergent Gauss's law relations. Usually, $Z_2$ charges in a stabilizer code can be defined consistently only if we choose a stabilizer basis. However, some choices of basis may result in conservation laws of the type in  \eqref{eq:gengl} that are not universal, i.e. do not appear for other basis choices.  
When the basis $\mc B$ is minimal, any alternative basis $\mc B^\prime$ generates a possibly larger non-local surface stabilizer group, i.e. $G_{\text{NLSS}} \subseteq G_{\text{NLSS}}^\prime$, up to equivalent choices of minimal basis (a more precise statement is discussed in Appendix \ref{appendixproofs}). Thus a minimal basis generates a universal set of Gauss laws.

Note that  both the stabilizer-counting and NLSS counting approaches to determining recoverable information involve counting stabilizers to determine the dimension of an appropriate stabilizer subspace, so that there is an occasional formal similarity between the two in practice. The reader is urged to keep in mind that the two are still quite different --- the former requires additionally the knowledge of entanglement entropy (e.g. by the methods of Ref.~\onlinecite{companion1}) whereas the latter gives the recoverable information directly and is therefore arguably more powerful. 

\section{Computations of Recoverable Information} \label{sec:compute}

\begin{figure}[t]
\centering
\centering
\includegraphics[width=.6\columnwidth]{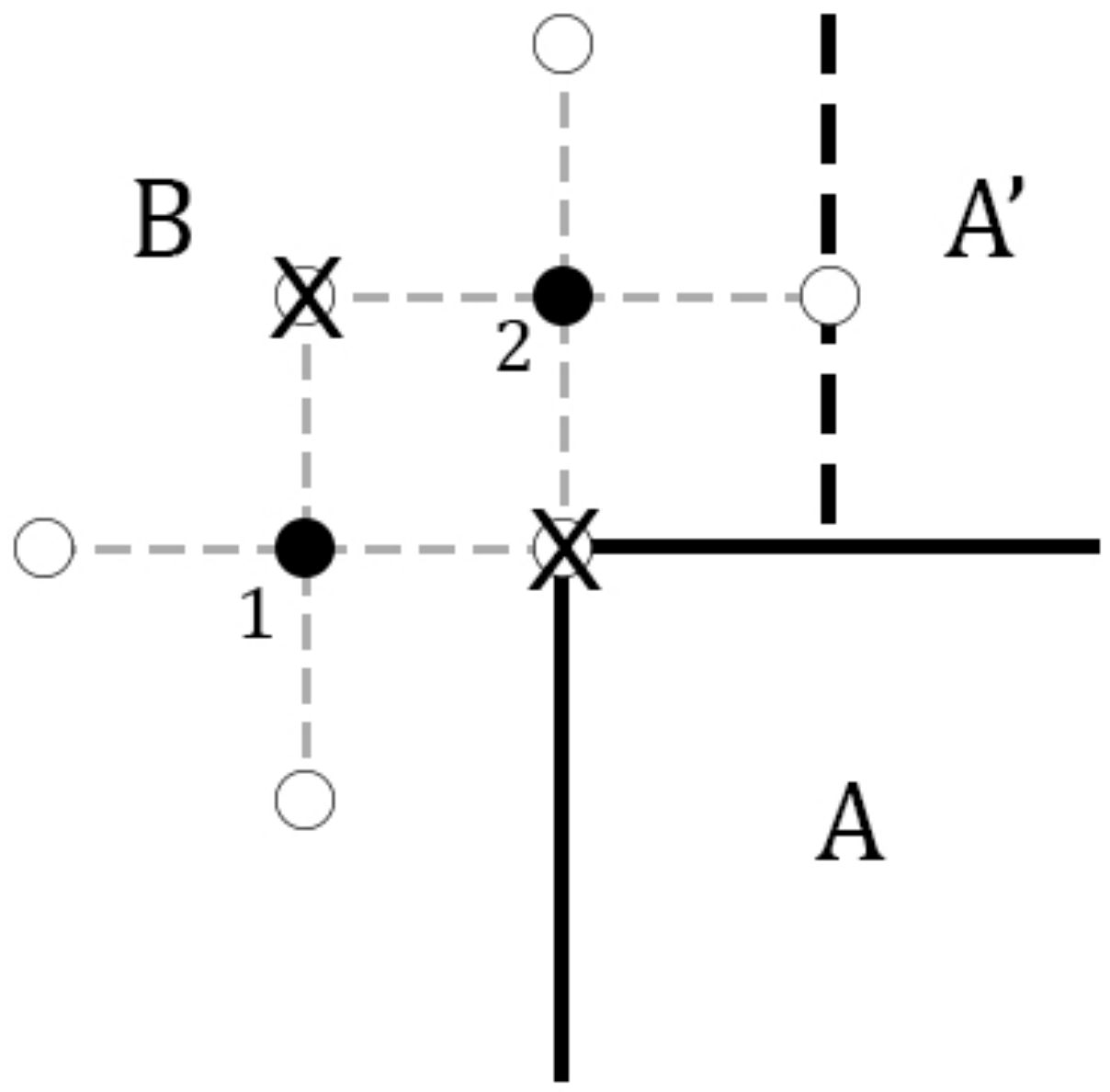}
\caption{Visualization of the `quasi-local NLSS' at the corner of a square entanglement cut. For simplicity, the black circles represent $Z$-type operators and the white circles represent the $X$-type operators.}\label{fig:cluster} 
\end{figure}%

\begin{figure}[t]
\centering
\centering
\includegraphics[width=.6\columnwidth]{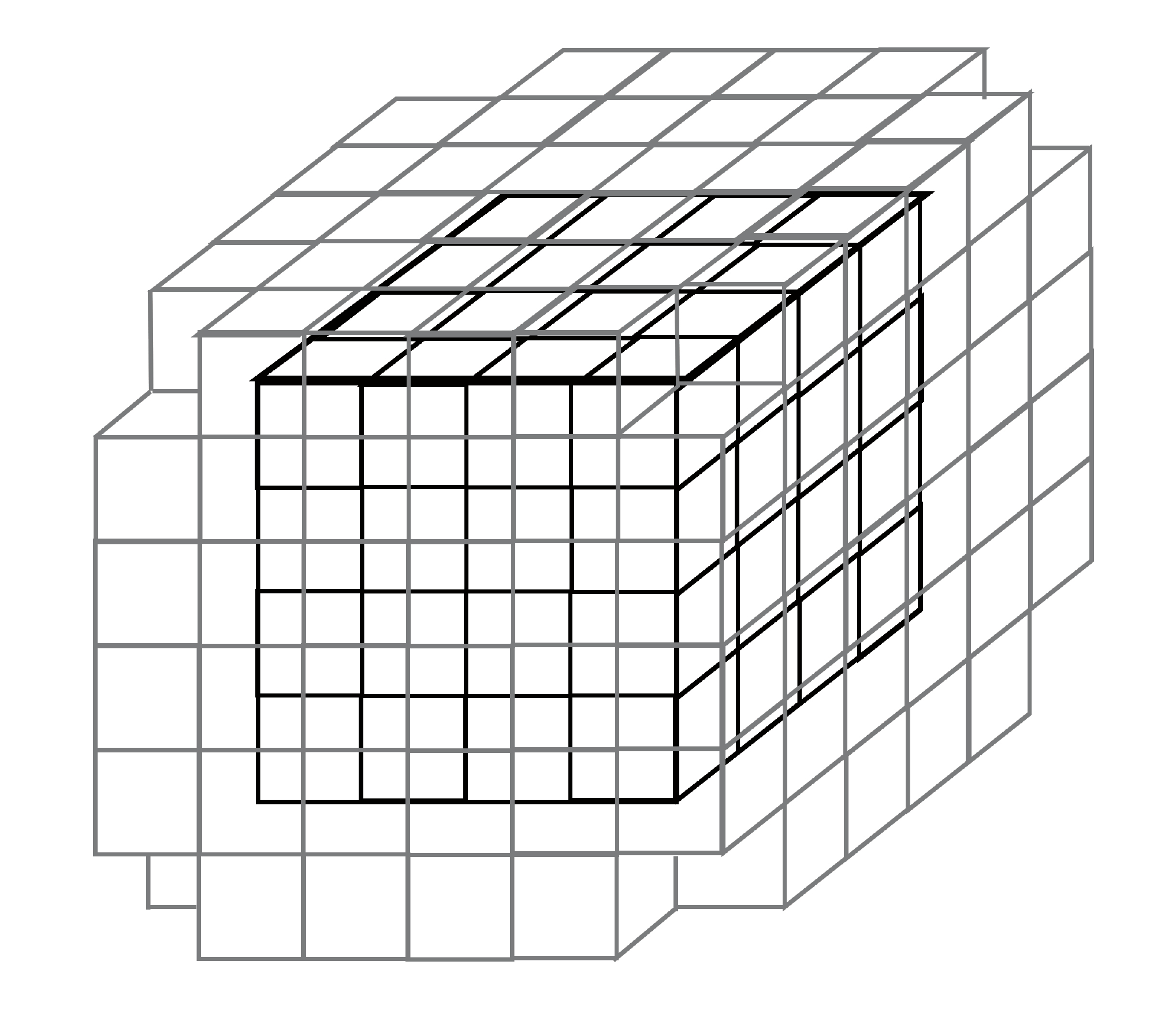}
\caption{Depiction of the constraint in the $d=3$ toric code which is formed by the product (sum in $W[\mc S]$) of all cut local cube constraints, but is itself not a cut constraint.}\label{fig:dented} 
\end{figure}%

We now apply the techniques of the preceding section to {four} different stabilizer codes introduced in Sec.~\ref{sec:models} : {the  cluster model in any dimension,} the $d=3$ toric code, the X-Cube model and Haah's code. Though first two are not fracton models, they demonstrate some of the necessary considerations involved in the minimization process used to define and compute the recoverable information. {The cluster models also serve  to illustrate the expected behavior for the recoverable information in SPT phases.} Note that we often identify elements of $G_{\text{NLSS}}$ by 
{ the element $g_B \in G_B$ for which there exist $g_A \in G_A$ and $g_\partial \in G_\partial$ such that $g_A g_B = g_\partial \in G_{\text{NLSS}}$. We do so as elements in $B$ are more intuitively thought of as NLSSs.} 

\subsection{Cluster model in arbitrary dimension $d$}\label{sec:computeclu}
{
We return to the cluster model, starting with $d=2$ and then arguing the same results hold in all other dimensions. We first take subsystem $A$ to be an  $R \times R$ square of vertices. The entropy is simply $S_A =4R-4$,  as there are $R^2$ qubits and $(R-1)^2$ complete stabilizers in $A$ (one for each internal vertex not along the edge). Turning to the cut stabilizers, note that every vertex on the boundary of $A$ represents one cut stabilizer which contributes $4(R-1)$ cut stabilizers. Similarly, almost every vertex in $B$ and on the boundary with $A$ is cut, {\it except} for the four vertex stabilizers at the corners; so vertices in $B$ contribute $4R$ to the count of cut stabilizers. Therefore, the recoverable information is
\begin{align}
\mu_{\text{clu}_2} (\text{square}) = 4R-4 + 4R - 2 (4R-4) = 4.
\end{align}
This may seem counterintuitive, as one would expect that the  recoverable information vanish, in accord with the fact that the cluster model does not exhibit topological order.} 
The extra bits of recoverable information can be attributed to the sharp corners of the entanglement cut. At a corner, we technically have a NLSS (though the name is not particularly apt in this case) coming from the product of two stabilizers oriented diagonally across a lattice plaquette and adjacent to the corner. Though each of these are cut by the corner, their product exactly removes support from that corner and as such is an NLSS (a product of cut stabilizers which is only supported in $B$; $g_A= I$ in this case). Such a situation is depicted in Fig. \ref{fig:cluster} for regions $A$ and $B$. However, if we choose the subregion to be $A\cup A'$ from Fig. \ref{fig:cluster}, stabilizer $2$ is now additionally cut along the boundary of $A'$. Then the product of stabilizers $1$ and $2$ is no longer entirely supported in $B$ and thus does not contribute to the recoverable information.  So instead of a square entanglement cut, if we choose a sufficiently `smooth' cut (i.e. one such that no corners are formed such as that shown in Fig. \ref{fig:cluster} for cut $A$), then we have zero recoverable information. This is analogous to topological entanglement entropy, where one only considers the universal constant piece of entanglement entropy for smooth boundaries. It should also be clear that this behavior is general in any dimension, where these `quasi-local NLSSs' appear in the recoverable information for sharp $d-2$ edges. For example in $d=3$, Fig.~\ref{fig:cluster} can be viewed as a 2D slice of a 3D boundary where the cluster stabilizers extend into and out of the page. So we conclude that in any dimensions $d$ and for a sufficiently smooth boundary, the recoverable information is 
\begin{align}
\mu_{\text{clu}_d}(\text{smooth}) = 0.
\end{align}
This result is the first demonstration of the power of independent NLSS counting; we can make very general statements without explicitly computing the entropy or the number of cut stabilizers. Furthermore, the fact that recoverable information is sensitive to sharp corners can be useful in certain cases --- e.g., the existence of corner entanglement can be indicate important universal features for certain models\cite{Bueno2015}  (also see discussion in Sec.~\ref{sec:conn}).

\subsection{$d=3$ toric code}\label{RIforTC3D}

\begin{figure}[t]
\centering
\centering
\includegraphics[width=.9\columnwidth]{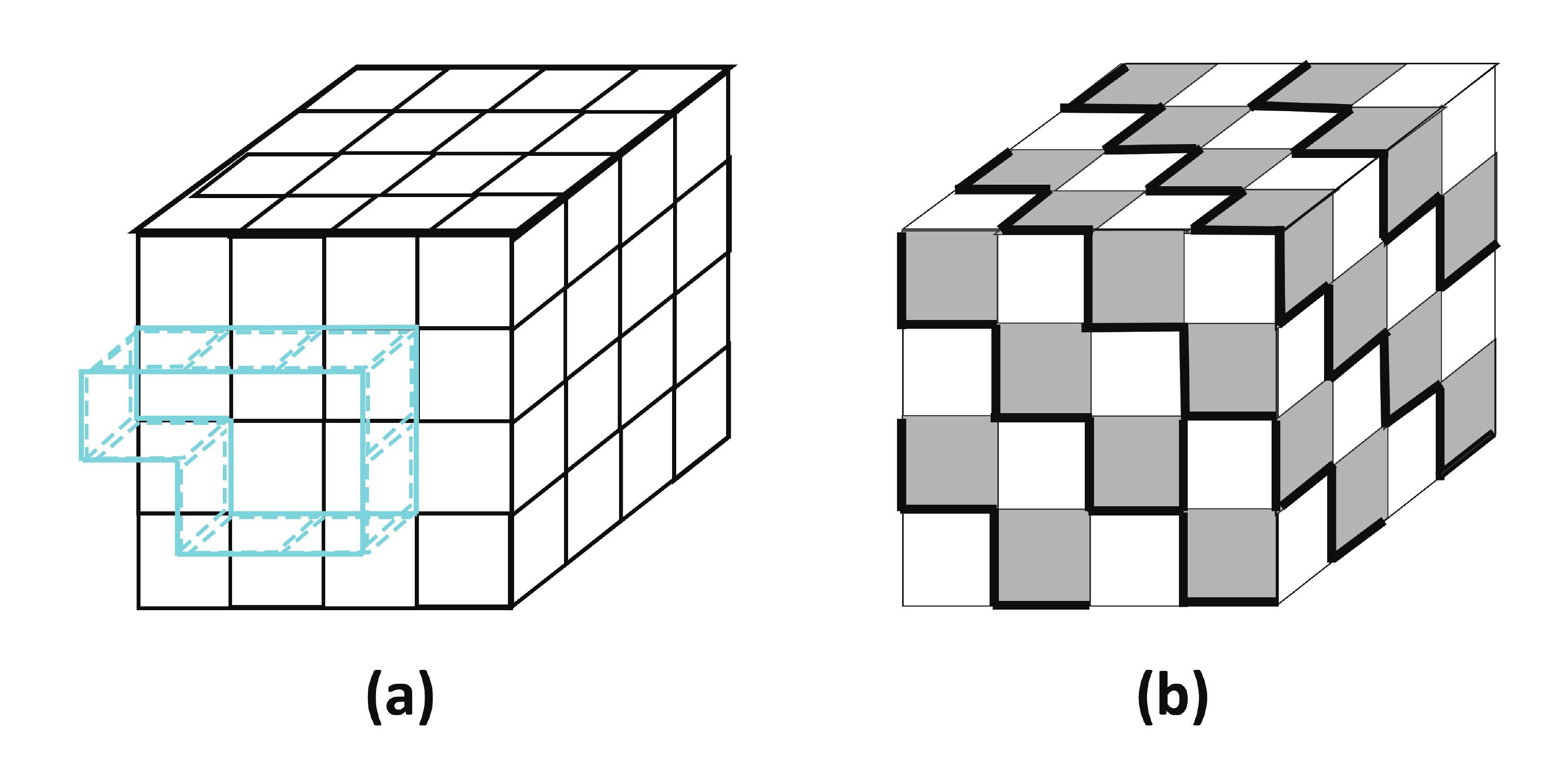}
\caption{Construction of an appropriate basis for the cut stabilizer group in the $d=3$ toric code. (a) Depiction of the straddling surface ribbons a minimal basis must avoid generating in $G_\partial$. (b) Basis construction using the checkerboard method as described in the text. Here, the thick lines indicate the intersection of a cut plaquette stabilizer in the basis with $A$.}\label{fig:3dtoricchecker} 
\end{figure}%

The $d=3$ toric code provides an example where counting cut stabilizers requires some caution to properly 
{account for} local constraints. We will take sub-region $A$ to be an $R\times R \times R$ cube (as measured by edge length on the lattice) with smooth surfaces. {Qubits on the entanglement surface are taken to lie `inside' the entanglement cut}. As in the $d=2$ toric code, we have the constraint that the product of all vertex stabilizers in the system is the identity. However, there is an extensive set of independent {\it local} constraints: namely, the set of all plaquette stabilizers around any cube (as their product is the identity). The latter reflects a redundancy in the information contained in the stabilizers: na\"ively, it seems as though the plaquette stabilizers on a single cube encode 6 bits of information, but thanks to this constraint, only 5 of these are independent. If such a cube is cut, Alice and Bob each receive one of two opposite faces, and between them they have two bits of information. The remaining $3$ bits are shared between the four cut stabilizers, and so if we simply counted cut stabilizers, we would {\it overestimate} the information loss when Alice and Bob are in the `no communication' state. This motivates the caveat on the definition of cut stabilizers and the cut stabilizer group as discussed in Section~\ref{sec:relation_recover}. 

We first compute $\mu$ by the cut-stabilizer-counting approach. According to the prescription given in Section~\ref{sec:relation_recover}, we can minimize the number of cut stabilizers by counting all cut stabilizers and subtracting the number of cut non-topological constraints. There are $6R^2+2$ cut vertex stabilizers (with no local constraints among these). There are $12 R(R+1)$ cut plaquette stabilizers, but there are $6R^2 + 12 R$ cut local {cubic} constraints involving plaquettes, 
one for each cube with an edge intersecting $A$.
Naively, subtracting these gives $6R^2$ for the contribution of plaquette stabilizers to $\mu$. However, under the map $f$ (defined in Sec. \ref{sec:const})  the product of all the cubic  constraints (alternatively understood as the sum of constraints  in $W[\mc S]$) is mapped to the product of all plaquettes in the boundary and all those in the layer of plaquettes just outside the boundary and parallel to it (excepting the three plaquettes that close each corner of the outer cube; see Fig. \ref{fig:dented}), and is therefore 
{\it not} cut. That is, there is no cut stabilizer to remove from this enveloping constraint in accordance with the minimization and cannot be subtracted from the total number of cut stabilizers.
  Note that this example neatly illustrates the importance of counting constraints with the group structure imposed (i.e., inside $W[\mc S]$) rather than simply taking unions of sets. Therefore, we find in effect that, for any choice of $\mc B_\partial$, there are $6R^2+1$ cut plaquette stabilizers, leading to $\min_{\mc B} [d_{\partial A}(\mc B)] = 6R^2+2 + 6R^2+1 = 12R^2+3$.

An alternative method to obtain this result is to recognize that no matter what basis we choose, it satisfies our requirements only if no set of cut plaquettes can be combined to form a closed ribbon operator whose parallel edges straddle the entanglement cut (i.e forms parallel strings in $A$ and $B$; see  Fig. \ref{fig:3dtoricchecker}(a)).These ribbons represent a cut local constraint {when one includes  the `caps' of the ribbon: the plaquettes on the boundary of $A$ encircled by the ribbon, and their parallel counterpart in $B$.}
 An algorithm for constructing a basis {free of such straddling ribbon operators} is as follows: first, we label the squares on the surface (not to be confused with plaquette stabilizers) in a ``checkerboard'' fashion. Then, we only include the cut plaquettes that intersect the ``black'' squares on two of the four edges, say the right and bottom. On the surface, this looks like several lines that stair-step diagonally across a face of the cube, as shown in Fig.~\ref{fig:3dtoricchecker}(b). At the intersection of two planes of the entanglement surface, we only include the plaquette perpendicular to the surface containing the black square. Careful consideration at those intersections shows that this does not allow closed ribbon loops. This implies that exactly half the edges on the boundary support a cut stabilizer in the basis, but we can add exactly one additional plaquette anywhere without creating a closed ribbon loop. Counting this way, we once again have $6R^2+1$ cut plaquettes and $6R^2+2$ cut vertex stabilizers (as the vertices satisfy no local constraints).  A simple counting argument using Eq. \eqref{eq:SAcount} yields $S_A = S_B=6R^2+1$, whence we find, using the definition \eqref{eq:fullmudef}, that 
\begin{align}
\mu_{\text{TC}_3}=1. 
\end{align}
 We can also compute the result in terms of the lower bound on $\mu$ 
as given by cut topological constraints. The set of all vertex stabilizers in the system satisfies a single topological constraint with pbcs, that is cut by the presence of $A$, leading to the bound $\mu\geq 1$, which is clearly saturated in this example.

  Finally, this also agrees with the counting in terms of non-local surface stabilizers: there is precisely one {NLSS}, namely the non-contractible membrane operator wrapping the cut  (product of all vertex stabilizers in $A\cup \partial A$).
There are no other NLSSs left after local constraints are implemented. In particular, to convince ourselves that no string operator can be in a minimal $G_{\text{NLSS}}$, the provided construction explicitly forbids creating cut ribbons which when combined with some element of $G_A$ would create a string. {One can readily check that such strings violate statements \ref{s2} and \ref{s3} and thus any basis which admits them cannot be minimal.}

 Note that the definition of $\mu$ also clearly satisfies our criterion of identifying the number of Gauss law-type conditions satisfied by the emergent degrees of freedom: it is well-known that the $d=3$ toric code can be identified with a $Z_2$ gauge theory, with a single such emergent conservation law for point-like particles (i.e. excitations with zero-dimensional support). The $Z_2$ magnetic fluxes in the 3D toric code form closed loops, rather than point-like particles as in $d=2$. The corresponding conservation law is linked to a closed-membrane rather than a closed-loop constraint of the Gauss-law type, and hence does not contribute to the recoverable information for simply connected entanglement cuts. 

\subsection{X-cube model}
\label{sec:xcube}

\begin{figure}[t]
\centering
\centering
\includegraphics[width=.6\columnwidth]{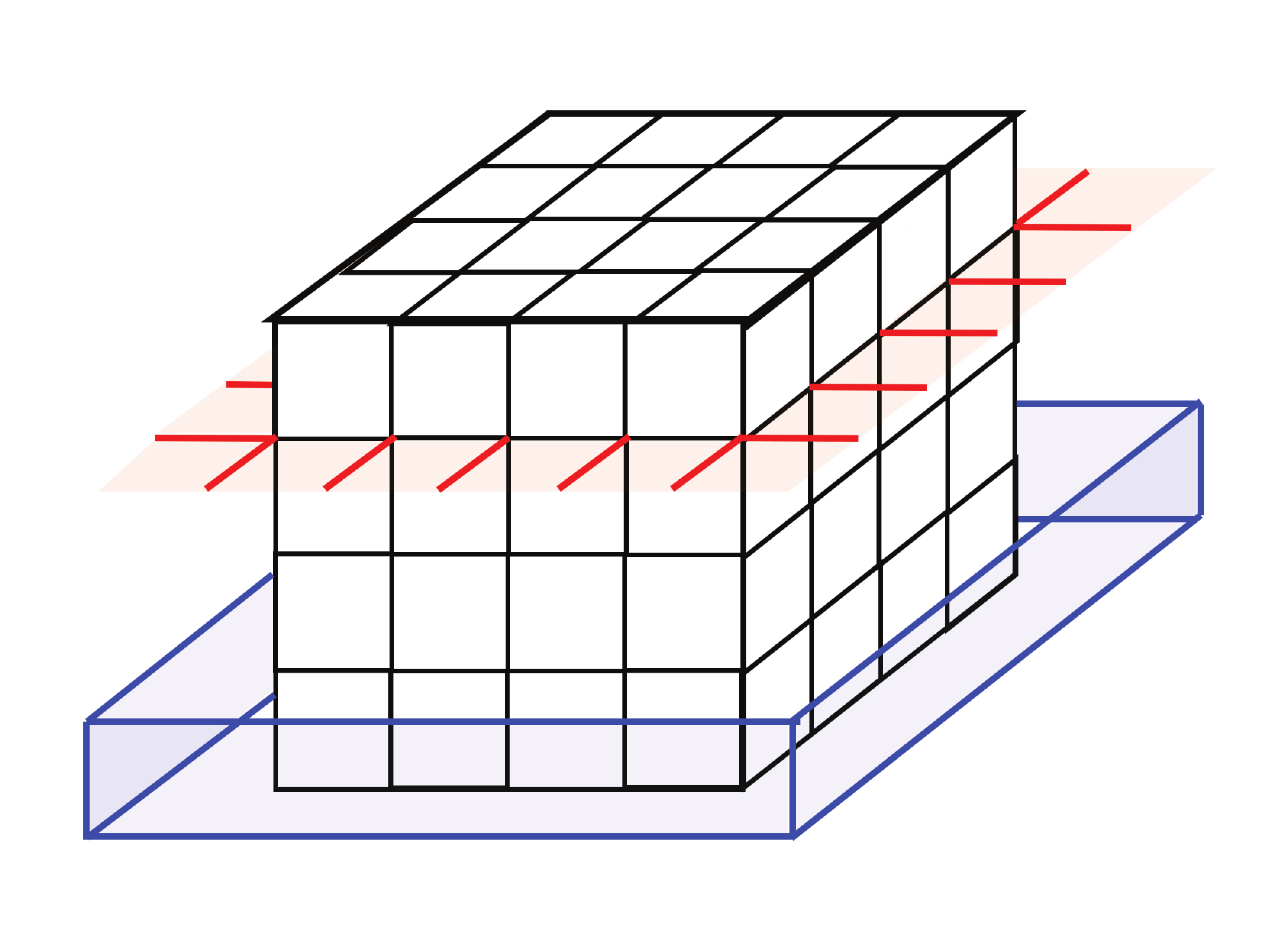}
\caption{Depiction of the NLSS ribbon operators for the X-cube model. The red operator is the $Z$-type ribbon operator generated from $A$-type stabilizers and the blue operator is the $X$-type ribbon operator generated from $B$-type operators.}\label{fig:Xribbon}  
\end{figure}%

We consider next the $X$-cube model, with the same choice of $A$ as for the $d=3$ toric code. As in that example,  
we have an extensive number of local (non-topological) constraints, namely the
set of all three $A$-type stabilizers for every vertex (since their product is the identity). However, these constraints are more easily handled as compared to  the $d=3$ toric code, as follows.
Clearly, when the six edges that coordinate a vertex are split between $A$ and $B$, at least two of the three associated vertex stabilizers are cut. When only two are cut, the minimization prescription dictates that one is removed from the basis, leaving one cut stabilizer; when all three are cut, we have no choice but remove one from the basis, leaving two cut stabilizers. 
By this counting, there are $6(R-1)^2$ cut $A$ type stabilizers on the faces (one for each vertex), $24(R-1)$ on the edges and 16 at the corners (two for each vertex). As for the $B$-type stabilizers, all the constraints are topological and every cube which intersects $A$ represents a cut $B$-type stabilizer. There are $6R^2$ intersecting cubes on the faces and $12R$ along the edges. Note that the  cubes that touch $A$ at its corners do not share an edge with it and are hence are not cut. So we find that $d_{\partial A} =12 R^2 + 24R-2$. Again, {a  counting argument using  Eq. \eqref{eq:SAcount} gives $S_A =S_B = 6R^2 +9R-4$, whence we find that 
\begin{align}
\mu_{\text{XC}} = 6(R+1).
\end{align}
}

The result is consistent with the number of cut topological constraints characterized by 
\bea
\prod_{v \in P_k} A_v^{(k)} &=& I \text{ and }\\
\prod_{c @ P_k} B_c &=& I,
\eea
where $P_k$ is the set of all vertices in a given plane perpendicular to the direction $k$, and we use the adjacency symbol (`$@$') to denote that the cube $c$ contains vertices in $P_k$. The constraints on the cubes are written in an apparently redundant manner since cubes are adjacent both above and below a given plane of vertices, and each plane of cubes (corresponding to a  dual lattice plane) contains an independent constraint. However, the constraints given above are still independent and are natural in terms of fracton excitation, in that this represents the subsystem in which bound pairs are mobile\cite{Prem2017}.
As each of these constraints is characterized by a plane and the constraint is cut if the plane is cut, we can see that the number of cut topological constraints is $6(R+1)$. 

Finally, we  can also interpret these results in terms of a natural basis for $G_{\text{NLSS}}$. From the $A$-type stabilizers, we can create non-contractible 
$Z$-type ribbon  
loops confined to a single layer of vertices (all vertex stabilizers entirely supported in a given layer of $A \cup \partial A$; 
 see Fig. \ref{fig:Xribbon}) and there are $3 (R+1)$ such operators. Meanwhile from the $B$-type stabilizers we can construct non-contractible X-type ribbon operators which are confined to a single plane of cubes (the product of all cube stabilizers in a layer of $A \cup \partial A$; see Fig. \ref{fig:Xribbon}). Note that these ribbons are supported on the joints where the ribbon changes direction. There are $3R$ such ribbons near the surface, but also $6$ additional ribbons which straddle the outer edge of any plane of the cut.  However, not all the operators we have identified above are truly independent. For the $Z$-type ribbons, consider multiplying all of them together. This is equivalent to multiplying all vertex stabilizers in $A$ and on the boundary which, put a different way, is the product of the three star stabilizers at every vertex in or on the surface of $A$. By the local vertex constraints, this is the identity, and we lose one NLSS from the basis. Furthermore, if we multiply all $X$-type ribbons wrapping around the cube in any one direction, we remove support where the ribbons overlap, but not on the outer edges or joints of the ribbons. Thus this operator is identical to having taken the same product in any other direction. This gives two additional constraint and we lose two more NLSSs from the basis. Altogether, there are $6(R+1)$ NLSSs in the basis after we implement the constraints, consistent with our earlier calculations for recoverable information.

The Gauss's law interpretation follows immediately from this basis for $G_{\text{NLSS}}$. Each $Z$-type ribbon yields a Gauss's law akin to that found in Eq. \eqref{eq:toricgl2} for the $d=2$ toric code, one for each of the $3R+3$ layers. Likewise, each $X$-type ribbon has an associated Gauss's law that again looks similar to the magnetic sector of the $d=2$ toric code, but within one of the $3R+9$ layers of cubes. However, while all of these Gauss law constraints are true, only $6(R+1)$ are independent in the following sense: given $6(R+1)$ of these conservation law statements, the other three can be deduced via direct correspondence with the discussion of constraints for $G_{\text{NLSS}}$.

\subsection{Haah's Code}
\label{subsec:recoverHaah}
We now turn to Haah's code, where we consider as sub-region $A$ a cube consisting of $R$ {\it vertices} on a side, always including both qubits at a vertex. Since the number of constraints is a complicated function of the boundary conditions, we conjecture that all of them are topological. Therefore, we need not be concerned about eliminating any of the cut stabilizers from ${\mc B}_\partial$. The number of cut stabilizers is therefore simply given by $d_{\partial A} = 2N_{\text{cubes}}(\partial A) -2$, where $N_{\text{cubes}}(\partial A)$ is the number of cubes wrapping the  entanglement cut (each of which has two types of stabilizers), and the subtraction accounts for the fact that exactly one of each type of stabilizer intersects $A$ only at its double identity corner. So, although the qubits on that corner are in $A$, the stabilizer formed in that cube by definition has no support in $A$ and therefore should not be counted as cut. To find $N_{\text{cubes}}(\partial A)$, we simply count the difference between the number of elementary cubes contained in an $(R+2) \times (R+2) \times (R+2)$ volume, less those in the $R\times R\times R$ interior volume. As these are $(R\pm 1)^3$, respectively, we find that $d_{\partial A} =2[(R+1)^3 -(R-1)^3] -2 = 12 R^2 +2$. Using Eq. \eqref{eq:SAcount}, we obtain $S_A = S_B =6R^2- 6R +2$ whence we find that 
\begin{align}
\mu_{\text{Haah}}= 12R -2.
\end{align} 

Before proceeding to compare this to the other methods for calculating recoverable information, we caution that there are some subtleties that need to be taken into account.
As mentioned above, the ground state degeneracy of Haah's code and hence the number of topological constraints is a complicated function of $L$:  while a certain sub-sequence of lattice sizes $L$ has a degeneracy exponential in $L$, there is a different sub-sequence of sizes where the degeneracy is independent of $L$.  Since the size of the entanglement cut $R$ can be any suitable fraction of the total system size $L$, the recoverable information can exceed the number of topological constraints. So not only is it difficult to count cut constraints, the bound they provide for recoverable information is far from tight. This emphasizes the fact that recoverable information is locally measurable and thus insensitive to the boundary conditions which dictate the number of global constraints and the ground state degeneracy.
 This distinction is hidden in more conventional topological orders where the local information in the entanglement is more directly linked to the global constraint and ground-state degeneracy properties.

We now wish to associate the recoverable information in Haah's code with a basis
{ for $G_{\text{NLSS}}$}; this is challenging owing to the rather complicated structure of the stabilizers in Haah's code. The counting is most easily achieved by means of the polynomial ring formalism described, e.g., in Ref.~\onlinecite{Vijay2015}. The full power of the polynomial ring formalism requires results in commutative algebra that are standard but possibly unfamiliar to our readers; however, we only require some basic ideas of this formalism, that we reproduce here for convenience.

The polynomial ring formalism provides a representation of translationally-invariant Hamiltonians built from products of Pauli operators in terms of `stabilizer maps', matrices whose entries are drawn from the Laurent polynomial ring  $\mathbb{F}_2(x,y,z)$ over the finite field on two elements (In simple terms, this is the set of Laurent polynomials in three variables with coefficients in the set $\{0,1\}$ where addition and multiplication are both taken modulo $2$).  As we describe below, this is merely an alternative representation of $V[P]$ and $V[G]$ which exploits polynomial algebra to represent translation symmetry.  Let us consider a model with a  {$q$}-site basis, i.e.  {$q$} qubits per unit cell: for example,  {$q =3$} for the X-cube model, while  {$q=2$} for Haah's code. The left dimension (number of rows) of the stabilizer map labels each of the distinct Pauli operators for each qubit in a unit cell.
As in our discussion of the vector space $V[P]$, we can represent $Y$-type operators by the product of the $X$- and $Z$-type operators, $Y =  iXZ$, so that the stabilizer map has  $2q$ rows. For example, we may take the first  {$q$} rows to correspond to $X$ on each basis site, and the next $q$ to $Z$. The right dimension  (number of columns), $m$ of the stabilizer map corresponds to the independent stabilizers associated to each unit cell. 
For instance, $m =3$ for the X-cube model {as each unit cell includes a single cubic stabilizer and two independent vertex stabilizers (since $A_v^{(xz)} = A_v^{(xy)}\cdot A_v^{(yz)} $). For Haah's code, $m=2$, corresponding to the two stabilizer types depicted in Fig. \ref{fig:Haah_code}. A stabilizer map is therefore a  {$2q \times m$} matrix.

The entries of this matrix elegantly capture the structure of the stabilizers, as follows. Each stabilizer can be associated to a given `origin' unit cell and involves some set of Pauli operators, $X$ or $Z$ acting on the $q$ qubits in some set of (usually adjacent) unit cells, linked by translation from the origin cell. Monomials in $d$ variables provide a convenient representation of the $d$-dimensional translation group: for instance, a lattice point $\mathbf{r} = a \hat{\mathbf{x}}+ b\hat{\mathbf{y}} +c\hat{\mathbf{z}}$ is maps to the monomial $x^a y^b z^c$. We may thus associate {a monomial in the matrix with the position (as measured from the origin cell) of a single qubit Pauli operator of a given type  ($X$ or $Z$) and index $q$ (that together specify the row within the matrix)  which makes up a given stabilizer type (that fixes the column). As in the case of $V[P]$, multiplication of Pauli operators which make up a stabilizer is represented by summation over the appropriate monomials. }
Translations of a stabilizer from the origin to an arbitrary point $(e.g., (a,b,c))$ on the lattice is then given by multiplication of the corresponding column by an appropriate monomial ($x^ay^bz^c$) which represents that position. 
{ Thus, by using arbitrary sums (which represent products) over individual monomials (which represent translations), we can generate all elements of the stabilizer group  by multiplying the resulting polynomial by a column of the stabilizer map and using distributivity. } 

In order to illustrate this and set the stage for our next calculation, let us apply this approach to Haah's code, where as we have noted {$q=m=2$}, so that the stabilizer map is a $4\times 2$ matrix. Let us focus on the $X$-type stabilizers $G_c^X$, corresponding to the first column of the stabilizer map. As such a stabilizer only involves the $X$-type Paulis, only the first two rows will be non-zero. Now, for a given site $\mathbf{r}$, $G_c^X$ involves the `right' qubit at the four sites $\{\mathbf{r}, \mathbf{r} +\hat{\mathbf{x}}, \mathbf{r} +\hat{\mathbf{y}}, \mathbf{r} +\hat{\mathbf{z}}\}$ and so the first row of the first column is the polynomial $1+x+y+z$, since we may, without loss of generality, choose the origin at $\mathbf{r}$. Similarly, $G_c^X$ involves the `left' cubit at $\{\mathbf{r}, \mathbf{r} +\hat{\mathbf{x}}+\hat{\mathbf{y}}, \mathbf{r} +\hat{\mathbf{y}}+\hat{\mathbf{z}}, \mathbf{r} +\hat{\mathbf{z}}+\hat{\mathbf{x}}\}$, so that the second row of the first column corresponds to $1+xy+yz+zx$. Proceeding analogously for $G_c^Z$, we find the stabilizer map for Haah's code,
\bea
S_{\text{Haah}} = \left( \begin{array}{cc} 	1 + x+y+z 		& 0 \\
								1+xy + yz+zx 	& 0 \\
									0 		& 1+\overline{xy} + \overline{yz}+\overline{zx} \\
									0 		& 1 + \overline{x}+ \overline{y} + \overline{z}
	 \end{array} \right),\nonumber\\
\eea
where $\overline{x}\equiv x^{-1}$, and we take the `origin' at the XX (ZZ) site (Fig.~\ref{fig:Haah_code}) for $G_c^X$ ($G_c^Z$).

This formalism allows us to count stabilizers with relative ease, by reducing the problem to one of counting independent polynomials. We focus on just the $X$-type stabilizers; the arguments carry over, {\it mutatis mutandis} to the $Z$-type ones. We define the pair of polynomials 
\be\label{eq:alphabetadef}
\alpha = 1 + x+y+z,\text{  }  \beta = 1+xy + yz+zx.
\ee
All other $X$-type stabilizers are obtained by simultaneous translations of these two polynomials, i.e. by multiplication by a monomial in the ring:  
 \be
\left(\begin{array}{c} \alpha\\ \beta\end{array}\right) \rightarrow x^ay^bz^c \left(\begin{array}{c} \alpha\\ \beta\end{array}\right).
\ee
 An arbitrary product of $X$-type stabilizers is then associated to an element of the polynomial ring via $g = P \cdot \left(\begin{array}{cc}\alpha &\beta\end{array}\right)^\top$, where $P\in \mathbb{F}_2(x,y,z)$. For instance, the product of $X$-type stabilizers located at $3\hat{\mathbf{x}} + 2 \hat{\mathbf{y}}$ and $\hat{\mathbf{y}} + 4 \hat{\mathbf{z}}$ is obtained by choosing $P = x^3y^2 + yz^4$. Similar relations hold for $Z$-type stabilizers, so that the entire stabilizer group is represented by elements of the polynomial ring.

For the cubic entanglement cut, we may represent every stabilizer group element contained in $A \cup \partial A $ as follows. First, take the origin (corresponding to $1$ in the ring) to be the lowest corner of the $(R+2)\times (R+2) \times (R+2) $ cube surrounding $A$; then, the elements correspond to those polynomials where any single variable has degree $\leq R+1$. 
We use this construction to find a dimension for $G_{\text{NLSS}}$ by noting that all elements of $\spn(\mc B_A) \times G_\partial$ (of $X$-type) can be generated by local stabilizers and thus are represented by the polynomial description above (see the Appendix of Ref.~\onlinecite{companion1} for a rigorous proof of this). We then impose the condition that such elements are also in $G_B$ (i.e., only supported in $B$) and count the number of independent solutions, thereby computing the dimension of $G_B \cap (\spn(\mc B_A) \times G_\partial)$. We follow this by arguing that $G_A \cap(\spn(\mc B_B) \times G_\partial)$ is trivial (i.e., our equations have no solutions), which completes our calculation of the dimension of the $X$-type stabilizers within $G_{\text{NLSS}}$ {via the relation \eqref{eq:NLSSfactor}} . Then, we simply draw on the inversion symmetry $\mathbf{r} \leftrightarrow -\mathbf{r}$ of Haah's code under $Z\leftrightarrow X$, while exchange the two qubits on the same site, to conclude that the answer for $Z$-type stabilizers is identical, so that the dimension of $G_{\text{NLSS}}$ is simply twice that of the $X$-type stabilizers alone.

First, if we write 
\be
P(x, y,z) = \sum_{\vec{a}} p(a_1, a_2, a_3) x^{a_1} y^{a_2} z^{a_3},
\ee
 where $p(a_1, a_2, a_3) \in \{0,1\}$ and the sum is over all integer vectors $\vec{a}$,  then we see that an arbitrary product of $X$-type stabilizers has the form
\be\label{eq:PolyRingexpansion}
P \cdot \left(\begin{array}{cc}\alpha &\beta\end{array}\right)^\top = \sum_{\vec{a} } \left(\begin{array}{cc}p_\alpha(\vec{a}) & p_\beta(\vec{a})\end{array}\right)^\top  x^{a_1} y^{a_2} z^{a_3},
\ee
where $p_\alpha(\vec{a}), p_\beta(\vec{a})$ follow from \eqref{eq:alphabetadef} by straightforward algebra.  Since the stabilizer is assumed to belong to $A \cup \partial A $, we assume that it has degree $\leq R+1$. If we want it to be supported entirely within $B$, we want no terms of degree $\leq R$ in {\it any} variable: in other words, for any term $\vec{a}$ in \eqref{eq:PolyRingexpansion} with $0< a_i \leq R$, we must have $p_\alpha(\vec{a}) = p_\beta(\vec{a})=0$. Using the form of $p_{\alpha, \beta}$ and the fact that $-p= p$ for $p\in \mathbb{F}_2$, we find a pair of conditions on the coefficients of $P$ ({one for each set of qubits}), when $0< {a}_i \leq R$,
\bea
p(\vec{a}) &=&   p(a_1-1, a_2, a_3)  +p(a_1, a_2-1, a_3) \nonumber\\ & & +p(a_1, a_2, a_3-1); \label{recur1}\\
p(\vec{a}) &=&  p(a_1-1, a_2-1, a_3) +p(a_1, a_2-1, a_3-1) \nonumber \\ & &  +p(a_1-1, a_2, a_3-1).\label{recur2}
\eea
Each of these relations is for one qubit at a given vertex, $\vec{a}$. As there are two equations for each vertex, the system appears to be over-determined. However, if we apply \eqref{recur2} to all terms on the RHS of \eqref{recur1} and vice-versa, we get the same equation, so that at a given vertex, the equations are redundant so long as $A$ includes the neighboring vertices on the three faces of the elementary cube bounding the negative octant. See Fig. \ref{fig:haahrecur} (a) for a visual representation of this. These {relations} severely limit the degrees of freedom which may be used to build non-local {surface} stabilizers. All vertices in the `positive octant' relative to the origin are fully fixed by the
{relations}, and thus not available. This potentially leaves vertices in the three boundary planes containing the origin. However even here, every vertex in the outer surfaces of $A$ adjacent to these planes no longer has the required neighbors to make \eqref{recur1} and \eqref{recur2} equivalent. To visualize this, picture the structure in Fig. \ref{fig:haahrecur}(a) ``puncturing'' the surface of $A$ so as to dangle into $B$. Whenever this happens, we gain one additional  {relation} and so to satisfy all equations, we must fix a coefficient in the boundary planes for every puncture. However, because the boundary planes `wrapping' $A$ have larger surface area than $A$, this leaves a few degrees of freedom left over, out of which we can build NLSSs. This can be seen in Fig, \ref{fig:haahrecur}(b) which depicts a 2D slice and the relevant punctures. The number of vertices in the boundary planes containing the origin is $3R^2 + 3R +1$, where the $3R^2$ comes from the faces, the $3R$ from the edges, and the $+1$ from the point at the origin. Meanwhile, the number of vertices fixed by the  
 adjacent surface layer of $A$
is $3 (R-1)^2 + 3(R-1) + 1$. This leaves $6R$ vertices that are not fixed by Eqs. \eqref{recur1} and \eqref{recur2}.
However, the origin corner does not count as a true boundary degree of freedom (because the stabilizer here contains a double identity, as discussed before). {Generalizing this argument, we realize that at the intersection of boundary planes (not at a corner), there are $3$ degrees of freedom to every one additional relation, thus giving 2 NLSSs per vertex in that intersection. Had we done the same analysis using stabilizers in $B$ and on the boundary, we would find that this is reversed: there are 3 relations for each degree of freedom. This implies that such a system is over determined, yielding no solutions, so we conclude that $G_A \cap (\spn(\mc B_B) \times G_\partial)$ is trivial. }
Therefore, including the $Z$-type stabilizers (by doubling the answer above) we find that the basis for $G_{\text{NLSS}}$ has $12R -2$ elements in agreement with our earlier results for recoverable information.

\begin{figure}[t]
\centering
\centering
\includegraphics[width=1\columnwidth]{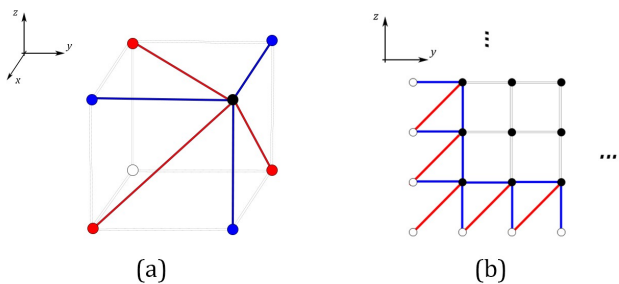}
\caption{Visual representation of Eqs. \protect \eqref{recur1} in blue and \protect \eqref{recur2} in red. (a) Representation for any vertex in $A$ (black vertex). (b) Representation of equations in a 2D slice near the intersection of two boundary planes (white vertices).}\label{fig:haahrecur}
\end{figure}%

We note  that the process of finding NLSSs is similar to that of finding constraints: we are looking for solutions which remove support from some subset of the space. 
While the topological degeneracy of Haah's code is a complicated function of system size, and admits of a simple answer only for certain system sizes, the largest topological degeneracy that is currently known to occur is  $\sim 4L$, for certain $L$. Our results suggest that $4L$ is not maximal, and that a degeneracy of $12L-2$ may be possible for certain boundary conditions, potentially including unequal lengths in each dimension. Using methods similar to those described above, it would be worth searching for the possible sequence of boundary conditions which saturate this upper bound. Such a sequence could be useful from the perspective of quantum error correcting codes as this would be a much denser coding than other models such as X-Cube or stacked toric codes.

We close this section by noting that while the Gauss's law picture for Haah's code is not presently clear, it would be an interesting topic for future work.

\section{\label{sec:conn} Discussion: Connections Between Recoverable Information, Non-local Entanglement, and Topological Properties}
Now that we have calculated recoverable information, we briefly analyze  the advantages and disadvantages of recoverable information over the non-local entanglement (NLE), given by multi-cut constructions such as those presented in \onlinecite{KitaevPreskill, LevinWen, GroverTurner, CastelnovoChamon2008, companion1}. The clear disadvantage is that we are restricted to stabilizer codes. However, when  a stabilizer framework exists, there is real power in using recoverable information. The first and most important is the ability to extract the structure of  Gauss-like laws associated to the boundaries of closed subsystems: every bit of recoverable information can be linked to an important piece of information about the excitations of the system. In contrast, as pointed out in Ref.~\onlinecite{companion1} (in particular, the section on Haah's code), NLE calculations indicate long-range entanglement and can possibly distinguish different phases, but there is no universal interpretation of the exact value in $d>3$ and the interpretation is quite complicated in $d=3$\cite{wen2017entanglement}. Additionally, the Gauss laws accessed via the recoverable information can further be interpreted as braiding or braiding-like rules for the excitations (e.g., Eq. \eqref{eq:toricgl2} or its generalization Eq. \eqref{eq:gengl}). The boundary operator can be broken into pieces, namely the cut stabilizers, each of which must generate excitations (or else they would have been in either the stabilizer or logical group). Thus applying these operators sequentially is  akin to braiding. For example, in the $d=2$ toric code, $(A_v)_A$ along the boundary is an operator which hops an electric particle. So applying the full boundary operator, $(g_\partial)_A$, is equivalent to creating a pair of electric excitations and threading one around the boundary to annihilate with the other. Then by the Gauss law, Eq. \ref{eq:toricgl2}, we know that if this process is sensitive to the parity of the number $N_m$ of magnetic excitations it encircles: the resulting phase is $-1^{N_m}$. A similar picture can be given for all other Gauss law-type relations given above. Finally, as recoverable information is defined for any bipartition (as long as it is less than the code distance), one can ask about other patterns in the entanglement. The most obvious example is sharp-edge contributions in any $d$ dimensional cluster model.  Another example is the $d=3$ toric code, where we conjecture that a more general formula for recoverable information is
\begin{align}
\mu_{TC_3} = 1+ 2g,
\end{align}
where $g$ is the genus of the boundary of $A$ (we are assume that the bipartition is connected). For such general bipartitions, loop NLSS operators can be included in a minimal $G_{\text{NLSS}}$ as they now satisfy statements \ref{s2} and \ref{s3}, so long as they wrap non-contractible loops. For an  example using the X-cube model, we consider more general but still zero-genus entanglement surfaces. Once again with the benefit of independent NLSS counting, one can see the only property of the surface which recoverable information depends on is the bounding rectangular parallelepiped. If the edge lengths of this parallelepiped are $R_1, R_2$ and $R_3$, then we conjecture the general formula for recoverable information is
\begin{align}
\mu_{\text{XC}}= 2( R_1 + R_2 + R_3 +3),
\end{align}
where we once again have all the same dependency among the ribbon operators as before. Though we have assumed a set of principal axes (a dependency we naturally expect) and a certain kind of scale dependence, recoverable information demonstrates that there is a deformation invariance in the entanglement structure, that is in a sense topological. Presumably, Haah's code has a similar deformation invariance though we leave a more rigorous treatment of this to future work.

In summary, recoverable information is sensitive to finer but often relevant features in the entanglement structure that can be lost to multi-cut constructions. So even though its use is limited, when possible, it can give us a more complete picture of the structure of entangled and topological states of matter.

\section{Concluding Remarks}\label{sec:conclude}

We have introduced a new concept - recoverable information - which is at present defined only for stabilizer codes, but which provides for a natural (albeit Hamiltonian dependent) characterization and interpretation of topological information, complementary to topological entanglement entropy. We have introduced three distinct ways of calculating recoverable information - via counting cut stabilizers, counting topological constraints, and counting non-local surface stabilizers. We have also illustrated these methods of calculating recoverable information for {the cluster model in arbitrary dimension, the $d=2$ and $3$ toric codes} and two archetypal `fracton' models - the X-cube model and Haah's code. {We argue that from recoverable information we may deduce the existence of $Z_2$ Gauss' law type constraints, which in turn suggests an interpretation in terms of emergent $Z_2$ conservation laws for point-like quasiparticles. }

While our present discussion of recoverable information is restricted to stabilizer codes, it is worth noting that small perturbations about stabilizer codes can be treated using the method of Schrieffer-Wolff transformations (as discussed in \cite{companion1}). This allows us to find `dressed stabilizer' operators, which are eigenoperators of the ground state, and may provide a (Hamiltonian dependent) route to defining recoverable information outside of the stabilizer limit.  This perspective appears similar in spirit to the Hamiltonian-dependent definition of ``dressed'' Wilson line operators by Hastings and Wen~\cite{Hastings}. Finally, note that for any given state, one can always find a (not necessarily local) Hamiltonian for which this is the ground state. In this respect, it may be possible to extract recoverable information directly from a wavefunction (perhaps by minimizing over inferred Hamiltonians), thereby relaxing the Hamiltonian dependence of the definition. We leave such extensions of the concept of recoverable information, as well as further applications of the formalism introduced herein, to future work.

\section*{Acknowledgements} We acknowledge extremely useful discussions and a previous collaboration (Ref.~\onlinecite{companion1}) with M. Hermele, and  thank B.~Ware for valuable discussions. We are also very grateful to Jeongwan Haah for valuable comments on the manuscript. H.M. is supported by the U.S. Department of Energy, Office of Science, Basic Energy Sciences (BES) under Award number DE-SC0014415. SAP and RMN acknowledge support from the the Foundational Questions Institute (fqxi.org; grant no. FQXi-RFP-1617) through their fund at the Silicon Valley Community Foundation. This work was also supported in part by the NSF under Grant No. DMR-1455366 (SAP) at the University of California, Irvine. 
\appendix
\section{Relating non-local surface stabilizers to recoverable information} \label{appendixproofs}

In this appendix, we prove the claims made in the main text that connect non-local surface stabilizer groups to recoverable information. For the following, it is useful to introduce some notational conventions in order to work conveniently with subsets of the stabilizer set. 
We denote the {\it symmetric difference} of  any two subsets $F, H \subseteq \mc S$, by $F \uplus H = F \cup H - F\cap H$. This is isomorphic to summation in $W[\mc S]$: we have
\begin{align}\label{eq:symdif}
\prod_{s\in F\uplus H} s &= \prod_{s \in F\cup H - F\cap H} s \nonumber\\&=  \prod_{s \in F - F\cap H} s \prod_{s'\in H - F\cap H} s' \prod_{s''\in F\cap H} (s'')^2 \nonumber\\&= 
\prod_{s\in F} s \prod_{s'\in H} s',
\end{align}
 since any elements in both $F\cap H$ will square to the identity in the RHS. (On a more technical level, any power set equipped with a symmetric product forms an Abelian group; all we have done is to show that this group is isomorphic to $W[\mc S]$.) This allows us to use both set-theoretic and  algebraic properties without constantly toggling between $W[\mc S]$ and $\mb P[\mc S]$. 

We will also require a few standard properties of the symmetric difference. First, we note that it is both commutative and associative --- these were already implicit in our statement that it lends an Abelian group structure to the power set. Another frequently-used relationship is that  $F \uplus H = F \cup H$ when the two $F,H$  are disjoint. Intersection on both right and left distributes over the symmetric difference, i.e. $J \cap(F \uplus  H) = (F \uplus H)\cap J = (J\cap F) \uplus (J \cap H)$; however,  set difference is only left distributive, i.e $ (F \uplus H)-J= (F-J)  \uplus  (H-J)$. Also, we will frequently switch back and forth between a stabilizer group element and a set of stabilizers that generates it. Then, for any $F\subseteq \mc S$,  we say $F \sim g \in G$ iff $\prod_{s\in F}s = g$. We use the same notation if two sets generate the same group member, i.e $F \sim H$ iff $\prod_{s \in F}s = \prod_{s \in H} s$. Clearly, this implies that if $F \sim H$ then $F \uplus H \in \mc C$. We also continue to employ  the convention that if a subset $F \subseteq \mc S$ includes a subscript of $A$, $B$ or $\partial$, then this represents the unique partitioning of that set into members only supported in $A$ ($F_A= \{ s\in F: \supp(s) \subseteq A\}$), members only supported in $B$ ($F_B = \{ s\in F: \supp(s) \subseteq B\}$), and members with support in both $A$ and $B$ ($F_\partial =\{s \in F: \supp(s) \cap A \neq \emptyset \text{ and } \supp(s) \cap B \neq \emptyset \}$). 

First we prove Eq. \eqref{eq:NLSSeq},  restated here for convenience:
\begin{align}\label{eq:simeq} 
G_{\text{NLSS}} \simeq G_A/ \spn(\mc B_A) \times G_B / \spn(\mc B_B).
\end{align}
\begin{proof}
We prove \eqref{eq:simeq} by constructing a linear (in the sense of $V[G]$) bijective function, $ \psi: G_{\text{NLSS}} \to G_A/ \spn(\mc B_A) \times G_B/ \spn(\mc B_B)$. 

Consider any $g_\partial \in G_{\text{NLSS}}$. We know that there exists unique elements $g_\partial \sim F_\partial \subseteq B_\partial$, $g_A \in G_A$, and $g_B \in G_B$ such that $g_\partial= g_A g_B$. Furthermore, there exists a pair of sets $H, J \subseteq \mc B$ such that $g_A \sim H= H_A \cup H_B \cup H_\partial$ and $g_B \sim J= J_A \cup J_B \cup J_\partial$. Define, $g_A^{n \mc B} \sim H_B \cup H_\partial$, and $g_B^{n \mc B} \sim J_A \cup J_\partial$. Clearly, $g_A^{n \mc B} \in G_A / \spn(\mc B_A)$ and $g_B^{n \mc B} \in G_B/ \spn(\mc B_B)$. As these sets uniquely represent these operators in the given basis, we define the map $g_\partial \mapsto \psi(g_\partial) = g_A^{n \mc B} g_B^{n \mc B}$, which makes $\psi$ injective. Clearly, this is also linear. 

Now consider any element of $ g_A \in G_A/ \spn(\mc B_A)$ and $g_B \in G_B/ \spn(\mc B_B)$. Each of these has a unique representation in our given basis, $g_A \sim  H= H_A \cup H_B \cup H_\partial\subseteq \mc B$ and $g_B \sim J= J_A \cup J_B \cup J_\partial \subseteq \mc B$, which by definition must be such that $H_A = J_B = \emptyset$. Using the stabilizer-generating equivalence we may write
$ H_\partial \sim g_A\prod_{s\in H_B}s$ and $J_\partial \sim g_B\prod_{s\in J_A}s$.
Rearranging, we have  
$H_\partial  \uplus  J_\partial \sim   (g_A \prod_{s \in J_A}s) (g_B \prod_{s\in H_B}s) \in G_A \times G_B$, which must also be in $G_\partial$ and therefore in $G_{\text{NLSS}}$. So we have that for all $g_A g_B \in G_A/ \spn(\mc B_A) \times G_B/ \spn(\mc B_B)$,  there exists a $g_\partial \in G_{\text{NLSS}}$ such that $\psi(g_\partial) = g_A g_B$. Therefore, $\psi$ is surjective and is thus a linear bijection. 
\end{proof}

It is a corollary of the above argument that
\begin{align}
G_{\text{NLSS}} \simeq & \left(G_A \cap (\spn(\mc B_B) \times G_\partial) \right) \nonumber \\
& \times  G_B \cap (\spn(\mc B_A) \times G_\partial), \label{eq:NLSSfactor}
\end{align}
 since $\left(G_A \cap (\spn(\mc B_B) \times G_\partial) \right) = G_A/ \spn(\mc B_A)$ and $\left(G_B \cap (\spn(\mc B_A) \times G_\partial) \right) = G_B/ \spn(\mc B_B)$.

Our ultimate goal is to prove the equivalence of \ref{s1} -\ref{s3}. First, however, we prove a pair of lemmas that are both necessary to this goal, as well as being illuminating in their own right. Before proceeding, we recall the standard result that if $G_1, G_2 \subseteq G$ are both subgroups of $G$, then $G_1 \cap G_2$ is also a subgroup of $G$; we will use this in what follows without explicit comment.

\begin{lemma}\label{relfactor} Given two subgroups $G_1$ and $G_2$ of a stabilizer group $G$ which we studied in this context, there exists a factorization for either group in the form $G_1 = (G_1 \cap G_2)\times [G_1 / (G_1\cap G_2)]$ or vice-versa, where $G_1/(G_1\cap G_2)$ is the {\it  quotient group} of $G_1$ with respect to $G_1\cap G_2$.\end{lemma}
\begin{proof}
This follows  from the fact that $G_1\cap G_2$ is a subgroup of both $G_1$ and $G_2$, and is clearly normal as both $G_{1,2}$ are abelian. Therefore, we can construct the quotient group $G_1 / (G_1\cap G_2)$ and the factorization is an immediate consequence. Of course, there is some ambiguity in this factorization: any $g \in G_1 / (G_1\cap G_2) $ is equivalent to any other member $g'$ of its coset, $g' \in g(G_1\cap G_2) \subseteq G_1$.
\end{proof} 

To remove ambiguity, we assume the factor $G_1 / (G_1\cap G_2)$ contains one and only one member of each coset, and in particular, $I$ from the identity coset $I(G_1 \cap G_2)$, so that $G_1 / (G_1\cap G_2)$ is a proper subgroup.

\begin{lemma}
(Basis Connection Lemma) For any basis $\mc B$, there exists a unique constraint basis $\mf C_{\mc B}$ (i.e., a basis for the constraint subspace $\mc C$), with the following properties:
\begin{enumerate}[label= (L2.\arabic*)]
 \item Every stabilizer not in the basis, $s \in \mc S -\mc B$, belongs to  a unique element $C_s$ of the constraint basis, $s\in C_s \in \mf C_{\mc B}$, and is in turn the unique non-basis stabilizer contained in $C_s$, i.e. $(C_s -\{s\}) \cap (\mc S - \mc B) = \emptyset$, \label{l1}

 \item Any constraint $C \in \mc C$ is uniquely represented in the basis $\mf C_{\mc B}$ by combining via symmetric difference the constraints $C_s \in \mf C_{\mc B}$  that correspond under  \ref{l1} to non-basis elements of $C$, i.e. $C= \biguplus_{ s \in C- \mc B} C_s$, and \label{l2}

\item Given any other basis, $\mc B '$, there exists a bijection $\psi: \mc B-\mc B^\prime \to \mc B^\prime- \mc B$, such that for all $s \in \mc B - \mc B^\prime$, we have $s, \psi(s)  \in C'_s \in \mf C_{\mc B'}$. \label{l3} 
\end{enumerate}
\end{lemma}
\noindent\ref{l3} gives this lemma its name. Stated simply,  for every element in one basis but not the other, we have a constraint which can be used to exchange that element with one from the other basis, thereby connecting the two bases.  
See Fig.~\ref{fig:bclemma} for a pictorial depiction of this result.

\begin{figure}[t]
\includegraphics[width=.5\columnwidth]{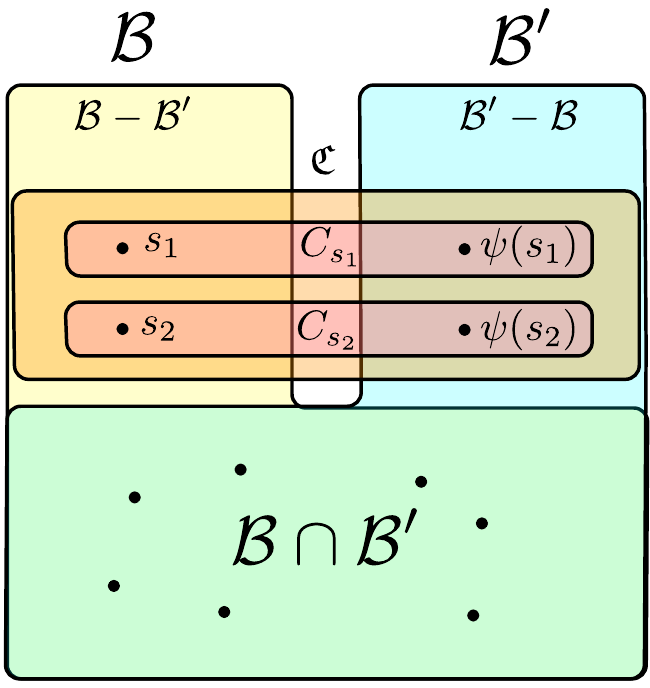}
\caption{\label{fig:bclemma}Pictorial depiction of the basis connection lemma, specifically \ref{l3}.}
\end{figure}
\begin{proof}
 Consider any element $s \in \mc S - \mc B$. As $s\notin \mc B$, and $\mc B$ is a basis, there must be a unique subset $F (s) \subseteq  \mc B$ such that $s \sim F(s)$.  
It then follows that $F(s)\cup \{s\} = C_s \in \mc C$, i.e. is a constraint. $C_s$ is uniquely specified by $s$ and clearly $(C_s -\{s\}) \cap (\mc S - \mc B) = \emptyset$. Therefore, if we define $\mf C_{\mc B}= \{C_s: s \in \mc S - \mc B\}$, this demonstrated \ref{l1}.

We now prove the set $\mf C_{\mc B}$ is a basis for $\mc C$, first by showing it is independent.  We do so by contradiction: assume that there exists some nonempty set $H \subseteq \mc S- \mc B$ such that 
 $\emptyset = \biguplus_{s\in H} C_s = \biguplus_{s \in H} (F (s) \cup \{s\}) = \biguplus_{s\in H} (F (s) \uplus \{s\})$,  where 
 in the second step we used the fact that $F(s) \uplus \{s\} =  F (s) \cup \{s\}$ since $F(s), \{s\}$ are disjoint. We then have, using the commutativity and associativity properties of the symmetric difference, that 
 \begin{align}
\emptyset &= \left[\biguplus_{s\in H} F(s)\right] \uplus \left[\biguplus_{s \in H}\{s\}\right]\nonumber\\ \implies& \biguplus_{s \in H} F (s) = \biguplus_{s \in H} \{s\} = H.\end{align}
This is a contradiction, since $H$ is nonempty by assumption and therefore must contain at least one element of $\mc S - \mc B$, whereas 
$F(s)$ contains no elements of $\mc S  -\mc B$ by \ref{l1}.
As well as being independent, the set $\mc C_{\mc B}$ also contains the exact number of constraints needed to span the constraint set ($|\mf C_{\mc B}| = |\mc S -\mc B|= |\mc S| - \dim(G) = \dim(\mc C)$ as discussed in Section \ref{sec:const}). Thus $\mf C_{\mc B}$ is a constraint basis for $\mc C$.

To prove \ref{l2}, Let $C \in \mc C$ be any constraint. As the constraint subspace is closed under symmetric difference, the following must be a constraint:
\begin{align}
\mc C &\ni \left\{\biguplus_{s \in C - \mc B} C_s\right\} \uplus C \nonumber\\&=\left\{\biguplus_{s \in C - \mc B} (F(s) \uplus \{s\})\right\} + C  \nonumber\\
&=\left\{\biguplus_{s \in C - \mc B} F(s)\right\} \uplus (C- \mc B)\uplus  C \nonumber \\
&= \left\{\biguplus_{s \in C - \mc B} F(s)\right\} \uplus (C\cap \mc B),
\label{eq:constraintC}
\end{align}
where we have again used the equivalence of the union and symmetric difference of disjoint sets as well as the disjoint partitioning of $C = (C\cap \mc B) \cup (C -\mc B) = (C\cap \mc B) \uplus (C -\mc B)$, which implies $(C- \mc B) \uplus C = C\cap \mc B$. Both terms in \eqref{eq:constraintC} are subsets of $\mc B$, again as $F(s) \subseteq \mc B$ by \ref{l1}, which means \eqref{eq:constraintC} is a constraint among members of $\mc B$. As $\mc B$ is a basis, there can be no nonempty constraints among its members, so it follows that  \eqref{eq:constraintC} is an empty constraint. Therefore  $\biguplus_{ s \in C- \mc B} C_s= C$, which is the unique representation of $C$ in the basis $\mf C_{\mc B}$.

Finally, we prove \ref{l3} by demonstrating the bijection $\psi$ exists. We first show that $\bigcup_{s \in \mc B- \mc B^\prime}( F^\prime(s)- \mc B) = \mc B^\prime - \mc B$, and then demonstrate that all the sets in the union are distinct and nonempty (the prime is included on $F'(s)$ to indicate that $F'(s) \subseteq \mc B'$, but is similarly defined by $s \sim F'(s)$). This immediately implies that we can construct a bijection that maps  $s \in \mc B-\mc B^\prime$ to some $\psi(s) \in F'(s) -{\mc B} \subseteq \mc B^\prime- \mc B$. Evidently,  $s, \psi(s) \in C'_s \in \mf C_{\mc B'}$.

As $F^\prime (s) \subseteq \mc B^\prime$, we have trivially that $\bigcup_{s \in \mc B- \mc B^\prime}( F^\prime(s)- \mc B) \subseteq \mc B^\prime - \mc B$. To show inclusion in the other direction, consider any $s^\prime\in \mc B^\prime - \mc B$. By \ref{l1}, we conclude that $s^\prime$ uniquely specifies a constraint $C_{s'} =F(s^\prime) \cup \{ s^\prime\}$ where $s' \sim F(s^\prime) \subseteq \mc B$. We also observe that  $F(s') - \mc B' \neq \emptyset$. (For, if this were not the case, then $F(s') \subseteq \mc B'$, but then $C_{s'} =F(s') \cup \{s'\}\subseteq\mc B'$ would be a nonempty constraint among members of  $\mc B'$, contradicting the fact that $\mc B'$ is a basis.) Then by $\ref{l2}$, the constraint $C_{s'}$ is represented in the $\mf C_{\mc B'}$ basis by $ C_{s'} = \sum_{s \in C_{s'}- \mc B'} C'_s = \sum_{s \in F(s')- \mc B'} C'_s$ (note we are using the fact that $s' \in \mc B'$ which implies $ C_{s'}- \mc B' = F(s')- \mc B'$). As $s' \in C_{s'}$, this implies that $s' \in C'_s - \mc B = F'(s) - \mc B$ for some $s \in  F(s')- \mc B'\subseteq \mc B - \mc B'$. Thus $\mc B^\prime- \mc B \subseteq \bigcup_{s \in \mc B- \mc B^\prime}( F^\prime(s)- \mc B) $, and we have proved inclusion in both directions, so that $\bigcup_{s \in \mc B- \mc B^\prime}( F^\prime(s)- \mc B) = \mc B^\prime - \mc B$.

 To show that no two sets of the form $F^\prime(s) - \mc B$ are equal, suppose instead that there exist $s, s^* \in \mc B$ such that $F'(s)- \mc B = F'(s^*) - \mc B$.  Using the disjoint partitioning of $F'(s) =(F'(s)\cap \mc B) \uplus (F'(s) - B)$, and similarly for $F'(s^*)$, we find that $\{s\} \uplus \{s^*\} \sim  (F'(s)\cap \mc B) \uplus  (F'(s^*)\cap \mc B)$. From this, we have  $\{s\} \uplus \{s^*\} \uplus (F'(s) \cap \mc B) \uplus (F'(s^*) \cap \mc B) \in \mc C$. As this constraint contains only members of a basis $\mc B$, it must be empty, and therefore from the uniqueness of $F'(s)$ and $F'(s^*)$, we must have $s=s^*$.
 Now $F'(s) - \mc B$ is nonempty for any $s$ by an argument exactly like that used above to show  $F(s') - \mc B' \neq \emptyset$. 
 
 From these results, we see that if we write $\bigcup_{s \in \mc B- \mc B^\prime}( F^\prime(s)- \mc B) = \mc B^\prime - \mc B$, then there are exactly as many distinct nonempty sets in the union on the LHS as there are elements in $\mc B' - \mc B$, and every member of $\mc B' - \mc B$ is contained in at least one of the sets in the union. Thus, we can bijectively map $s$ onto   $\psi(s) \in F'(s) - \mc B \subseteq \mc B' - \mc B$, in which case $s, \psi(s) \in F'(s) \cup\{s\} = C'_s\in \mf C_{\mc B'}$.
\end{proof}

One consequence of the basis connection lemma is that in minimizing over bases to find $\dim(G_{\text{NLSS}})=\mu$, there are no `local minima'. That is, every basis is directly connected to a minimal basis via some independent constraint set,  and so the minimizing procedure described in the main text will eventually reach this minimal basis, which is necessary  in order for it to reproduce the recoverable information $\mu$.

With these lemmas, we are ready to prove the equivalence of \ref{s1}-\ref{s3}, reproduced here for convenience:

\begin{enumerate}[label= (S\arabic*)]

\item $\dim( G_{\text{NLSS}}) = \mu$, \label{s1}

\item For any element $g_\partial \in G_{\text{NLSS}}$ with its associated $g_A \in G_A$ and $g_B \in G_B$ for which $g_\partial = g_A g_B$, if there exist sets of {stabilizers}, {$F_A \subseteq \mc S_A$ and $F_B \subseteq \mc S_B$  
 whose product is $g_A, g_B$ respectively}, then $F_A \cup F_B \cup F_\partial \in \mc T$, where $F_\partial \subseteq \mc B _\partial$ is the basis set which generates $g_\partial$ and \label{s2}

\item For all $g_\partial  \neq I \in G_{\text{NLSS}}$ such that $F_\partial \subseteq \mc B _\partial$ is the basis set which generates $g_\partial$, if there exists a {non-topological} constraint $C \in \mc C - \mc T$ such that $F_\partial \subseteq C$, then $C_\partial - F_\partial \neq \emptyset$. \label{s3}

\end{enumerate}

\begin{proof}
Without loss of generality, let $\mc S$ be such that $ \mc T = \{\emptyset\}$, as we can always change boundary conditions so that this is the case without changing $\mu$. Let $G_{\text{NLSS}}$ represent the non-local surface stabilizer group generated by the basis $\mc B$.

\noindent{\bf \ref{s1} $\implies$ \ref{s2} (by contraposition).} For the sake of argument, suppose there exists a $g_A g_B =g_\partial \in G_{\text{NLSS}}$ with $g_\partial \sim F_\partial \subseteq \mc B_\partial$, $g_A \in G_A$ and $g_B \in G_B$ such that there exist sets $F_A\subseteq \mc S_A$ and $F_B \subseteq \mc S_B$ with $g_A\sim F_A$ and $g_B \sim F_B$.  Then, since $g_\partial=g_A g_B $ implies that $F_\partial \sim F_A \uplus F_B$, we have that $F_A \uplus F_B \uplus F_\partial = C \in \mc C$. Now, suppose this constraint is not topological, i.e. $C \not\in \mc T$; since $\emptyset \in \mc T$, it follows that $C\neq \emptyset$. From \ref{l2}, we have that the expansion of $C$ in the $\mf C_{\mc B}$ basis is given by $C= \biguplus_{s \in C- \mc B} C_s$, where $C_s \in \mf C_{\mc B}$ is the unique constraint associated with $s \in C- \mc B$. Let $s_\partial$ be any member of  $F_\partial \subseteq C$. According to the expansion of $C$, $s_\partial \in C_s$ for some $s \in C- \mc B$, which is clearly not a member of the basis. Also by \ref{l1}, $C_s -\{s\} \subseteq \mc B$. With this, consider the set $\mc B^\prime = (\mc B -\{s_\partial\}) \cup \{s\}$. As any member of $\mc B$ is a product of members from $\mc B^\prime$ (namely $s_\partial \sim C_s -\{s_\partial\} \subseteq \mc B'$), this is a basis for which $F_\partial \nsubseteq \mc B^\prime$ and we have that $\mc B^\prime$ generates $ G_{\text{NLSS}}^\prime \subsetneq G_{\text{NLSS}}$. Therefore, $\mc B$ is not minimal and by contraposition,  \ref{s1} implies \ref{s2}.

\noindent{\bf \ref{s2} $\implies$ \ref{s3}.} Suppose $G_{\text{NLSS}}$ satisfies \ref{s2}. Now suppose there exists $ g_\partial \neq I \in G_{\text{NLSS}}$ such that $g_\partial \sim F_\partial \subseteq C \in \mc C$, where $F_\partial \subseteq \mc B_\partial$.  It follows then that we can also write $g_\partial \sim  C - F_\partial$. Using the standard  decomposition of set $C$, we may rewrite this as $g_\partial \sim C_A \cup C_B \cup (C_\partial -F_\partial)$. Now, suppose, contrary to  \ref{s3}, that  $C_\partial -F_\partial = \emptyset$. Then, we have $g_\partial \sim C_A \cup C_B$, from which it follows that $g_\partial =g_Ag_B$ with $G_A \ni g_A \sim C_A$ and $ G_B \ni g_B \sim C_B$. Now,  using \ref{s2} we  have that $C\in \mc T$, but since we also have $C\neq \emptyset$, this contradicts  $\mc T = \{\emptyset\}$ (as we have assumed at the outset). So, it must be that $C_\partial - F_\partial \neq \emptyset$, which proves \ref{s3}.

\noindent{\bf \ref{s3} $\implies$ \ref{s1}.} Finally, suppose $G_{\text{NLSS}}$ satisfies \ref{s3}, and  let $\mc B^\prime$ be any other basis which generates the non-local surface stabilizer group $G_{\text{NLSS}}^\prime$. We now construct an injective map $\chi:G_{\text{NLSS}} \rightarrow G_{\text{NLSS}}'$; this immediately requires that 
  $\dim(G_{\text{NLSS}}) \leq \dim(G_{\text{NLSS}}')$ and hence we can conclude that the basis $\mc B$ is minimal.

First, we observe that $G_{\text{NLSS}}$ and $G_{\text{NLSS}}'$ are both subgroups of $G$. Using Lemma~\ref{relfactor}, and defining $G_{\cap} =G_{\text{NLSS}} \cap  G_{\text{NLSS}}^\prime$ we may factor $G_{\text{NLSS}} = G_{\cap} \times [G_{\text{NLSS}} /G_{\cap}]$. We now show that, given any $g_\partial \neq I \in G_{\text{NLSS}} /G_{\cap}$ with $g_\partial \sim F_\partial \subseteq\mc B_\partial$,    $F_\partial - \mc B^\prime$ is uniquely specified by $g_\partial$. To see this, let  $g_{\partial 1} \sim F_{\partial 1}\subseteq \mc B_\partial$ be any member of $G_{\text{NLSS}}/G_\cap$ such that $F_{\partial 1} - \mc B^\prime = F_\partial  - \mc B^\prime$, and consider the disjoint partition $F_\partial =( F_\partial \cap \mc B') \uplus (F_\partial - \mc B')$, and similarly for $F_{\partial 1}$. Then $g_\partial g_{\partial 1} \sim F_\partial \cap \mc B^\prime \uplus F_{\partial 1} \cap \mc B^\prime$. Thus, $g_\partial g_{\partial 1}$ is generated by some subset of  $\mc B_\partial '$, and is therefore contained in $G_{\text{NLSS}}'$; but, since $g_\partial, g_{\partial1}\in \mc B_\partial$, they must both also be contained in $G_{\text{NLSS}}$. From this, we have that $g_\partial g_{\partial1} \in G_\cap$, but as  the factor $G_{\text{NLSS}} /G_{\cap}$ is closed and only trivially intersects $G_\cap$, it follows that $g_\partial =g_{\partial 1}$, i.e. $F_\partial - \mc B^\prime$ is uniquely specified by 
 $g_\partial$ (for the given factorization of $G_{\text{NLSS}}$ or, equivalently, unique to the coset of $g_\partial$).

 By \ref{l1}, for every $s \in F_\partial - \mc B^\prime$ there exists a unique $C'_s \in \mf C_{\mc B'}$ such that $s \in C'_s$ and $(C'_s -\{s\}) \cap (F_\partial - \mc B^\prime) = \emptyset$. As each member of $F_\partial - \mc B^\prime$ is in one and only one of these constraints, we have $F_\partial - \mc B^\prime \subseteq \biguplus_{s \in F_\partial - \mc B^\prime} C'_s = C \in \mc C$. 
 Furthermore, we have 
 \begin{align}
 C- (F_\partial - \mc B') &=\biguplus_{s \in F_\partial - \mc B^\prime} C'_s - (F_\partial - \mc B') \nonumber\\&= \biguplus _{s \in F_\partial - \mc B^\prime} (C'_s -\{s\}) \subseteq \mc B', 
 \end{align} where we have used the left-distributivity of set difference over symmetric difference and the fact that $F_\partial -\mc B' \subseteq \mc B'$ with $(F_\partial -\mc B') \cap C'_s = \{s\}$. Note that the  constraint $C$ is unique to $F_\partial - \mc B^\prime$ and thus unique to $g_\partial \in G_{\text{NLSS}}/G_{\cap}$ 
. So, we have that $F_\partial - \mc B^\prime \sim C- (F_\partial - \mc B^\prime) = C_A \uplus C_B \uplus \left(C_\partial - (F_\partial - \mc B^\prime)\right)$ (as $C_{A(B)}\cap (F_\partial - \mc B^\prime) = \emptyset$ by definition). Using this result and the disjoint partitioning of $F_\partial$, we have
\begin{align}\label{eq:appinter1}
g_\partial \sim F_\partial =& (F_\partial \cap \mc B^\prime) \uplus (F_\partial - \mc B^\prime)\nonumber \\
\sim& 
C_A \uplus C_B \uplus  \left[\left\{F_\partial \cap \mc B^\prime\right\}\uplus \left\{C_\partial - (F_\partial - \mc B^\prime)\right\}\right].
\end{align}
But we also have that $g_\partial = g_A g_B$ for some $g_A \in G_A$ and $g_B \in G_B$. This implies that  %
\begin{align}
  \left\{F_\partial \cap \mc B^\prime\right\}&\uplus \left\{C_\partial - (F_\partial - \mc B^\prime)\right\}
 \nonumber \\
\sim& \left(g_A \prod_{s \in C_A}s \right)\left(g_B \prod_{s \in C_B}s \right) \in G_{\text{NLSS}}^\prime,
\end{align}
where in the second line, we use $(F_\partial \cap \mc B^\prime) \uplus \left(C_\partial - (F_\partial - \mc B^\prime)\right) \subseteq \mc B'$ to argue that we have an element in $G_{\text{NLSS}}^\prime$. 
From this, we can define a map $\chi :G_{\text{NLSS}} \to G_{\text{NLSS}}^\prime$ such that  
 \be
 g_\partial \mapsto \chi(g_\partial) \sim (F_\partial \cap \mc B^\prime) \uplus \left(C_\partial - (F_\partial - \mc B^\prime)\right).\label{eq:chidef}
 \ee 
 Note that for any $g_\partial \in G_\cap$ with $g_\partial \sim F_\partial$, $F_\partial - \mc B^\prime$ is empty by definition; in this case it generates an empty constraint via the preceding construction. For such an element, $\chi(g_\partial)\sim F_\partial \cap \mc B' = F_\partial \sim g_\partial$.

 We now show that $\chi$ is linear.  To do this, consider any $g_{\partial1}, g_{\partial2} \in G_{\text{NLSS}}$, with $g_{\partial1} \sim F_{\partial1} \subseteq \mc B$ and $g_{\partial2} \sim F_{\partial2} \subseteq \mc B$. This implies $g_{\partial1} g_{\partial2} \sim F_{\partial 1} \uplus F_{\partial 2} \subseteq \mc B$, which is its unique basis representation. By our previous arguments, we have the constraint $C_{12} \supseteq(F_{\partial 1} \uplus F_{\partial 2}) - \mc B^\prime$ uniquely associated with $  (F_{\partial 1} \uplus F_{\partial 2}) - \mc B'$. Moreover
\begin{align}
  C_{12} &= \biguplus_{s\in (F_{\partial 1} + F_{\partial 2}) - \mc B^\prime} C_s\nonumber \\
 &= \left\{\biguplus_{s \in F_{\partial1} - \mc B'} C_s\right\} \uplus \left\{ \biguplus_{s \in F_{\partial2} - \mc B'} C_s\right\} \nonumber \\
 &= C_1 \uplus C_2,
\end{align}
 where $C_{1(2)}$ is the unique constraint for $F_{\partial 1(2)}- \mc B'$ as discussed above. Here we have used the fact that for $s \in (F_{\partial 1} \cap F_{\partial 2}) - \mc B^\prime$, $C_s$ cancels out in the expansion of $C_1 \uplus C_2$, allowing us to combine the two sums into the one, yielding $C_{12}$. Using this relation, we have $F_{\partial1}- \mc B' \sim C_1 -(F_{\partial1}- \mc B')$,  $F_{\partial2}- \mc B' \sim C_2 -(F_{\partial2}- \mc B')$ and  $(F_{\partial1}\uplus F_{\partial2}) - \mc B' \sim (C_1\uplus C_2) - \left((F_{\partial1}\uplus F_{\partial2}) - \mc B'\right)$. This in turn implies
\begin{align}\label{eq:intapp}
 (F_{\partial1}\uplus F_{\partial2}) - \mc B' =&(F_{\partial1}- \mc B')\uplus( F_{\partial2}- \mc B') \nonumber \\
\sim& \left(C_1 -(F_{\partial1}- \mc B')\right)\uplus  \left(C_2 -(F_{\partial2}- \mc B')\right) \nonumber \\ 
 \sim & (C_1\uplus C_2) - \left((F_{\partial1}\uplus F_{\partial2}) - \mc B'\right).
\end{align}
Since both sets on either side of the last $\sim$ relation are subsets of $\mc B'$, and any constraints among members of $\mc B'$ must be empty, we have that the last  $\sim$ of Eq. \eqref{eq:intapp} can be replaced by set equality.  
To further simplify this equality, we recognize $C_{1(2)} = (C_{1(2)})_A \uplus (C_{1(2)})_B\uplus (C_{1(2)})_\partial$ and observe all sets being subtracted from these constraints in Eq. \eqref{eq:intapp} contain only cut stabilizers. In that case, $(C_{1(2)})_A \uplus (C_{1(2)})_B$ is unaffected and can be canceled  from both sides of Eq. \eqref{eq:intapp} so that
\begin{align}\label{eq:intapp2}
 &\left((C_1)_\partial -(F_{\partial1}- \mc B')\right)\uplus  \left((C_2)_\partial -(F_{\partial2}- \mc B')\right)\nonumber \\
&= ((C_1)_\partial \uplus(C_2)_\partial) - \left((F_{\partial1}\uplus F_{\partial2}) - \mc B'\right).
\end{align}
 We now compute the action of $\chi$ on the product $g_{\partial1} g_{\partial2}$ using the definition (\ref{eq:chidef}):
\begin{align}
\chi(g_{\partial1} g_{\partial2}) \sim& [(F_{\partial1} \uplus F_{\partial2}) \cap \mc B']\uplus \left[\left\{(C_1)_\partial +(C_2)_\partial\right\} \right. \nonumber\\
&- \left. \left\{(F_{\partial1} \uplus F_{\partial2}) - \mc B'\right\}\right]\nonumber\\
 = & (F_{\partial1} \cap \mc B') \uplus\left[ (C_1)_\partial -(F_{\partial1}- \mc B')\right] \nonumber\\
& \uplus (F_{\partial2} \cap \mc B') \uplus\left[ (C_2)_\partial -(F_{\partial2}- \mc B')\right] \nonumber \\
\sim& \chi(g_{\partial1}) \chi(g_{\partial2}),
\end{align}
where we have used Eq. \eqref{eq:intapp2} and distributivity of intersection over symmetric difference. This completes the argument that $\chi$ is linear.

Up to this point, we have yet to use the fact that $\mc B$ and $G_{\text{NLSS}}$ satisfy \ref{s3}. However, consider some $g_\partial \in G_{\text{NLSS}}$ with $g_\partial \sim F_\partial \subseteq \mc B$ and suppose $ \chi(g_\partial) = I$. Using the unique constraint $C\supseteq F_\partial - \mc B'$ associated with $F_\partial -\mc B'$, we can apply the definition of $\chi$, which implies $F_\partial \cap \mc B^\prime \uplus \left(C_\partial - (F_\partial - \mc B^\prime)\right) \in \mc C$. But all members of this constraint are members of $\mc B'$, which implies this constraint is empty. Thus $F_\partial \cap \mc B^\prime = \left(C_\partial - (F_\partial - \mc B^\prime)\right)$. Using the disjoint partitioning of $F_\partial$, this implies $F_\partial = (F_\partial \cap \mc B' )\cup (F_\partial - \mc B') =\left(C_\partial - (F_\partial - \mc B^\prime)\right)\cup (F_\partial - \mc B') \subseteq C$.  With this, we consider 
\begin{align}
C_\partial - F_\partial =& C_\partial - \left((F_\partial \cap \mc B' )\cup (F_\partial - \mc B')\right)\nonumber \\
 =& \left( C_\partial - (F_\partial \cap \mc B')\right) \cap  \left(C_\partial - (F_\partial - \mc B')\right) \nonumber \\
=& \left( C_\partial - (F_\partial \cap \mc B')\right) \cap (F_\partial \cap \mc B^\prime) = \emptyset,
\end{align}
using $F_\partial \cap \mc B^\prime = \left(C_\partial - (F_\partial - \mc B^\prime)\right)$ and DeMorgan's rule for distributing set difference over union in the second line. But by our hypothesis, such a set can only be empty for $g_\partial =I$. Thus only $\chi(I) =I$, which as $\chi$ is linear, implies $\chi$ is injective (i.e has a trivial kernel). 

Injectivity implies that the imagine of $G_{\text{NLSS}}$ under $\chi$ has the same dimension as $G_{\text{NLSS}}$ itself. As $\chi[G_{\text{NLSS}}] \subseteq G_{\text{NLSS}}^\prime$, we therefore have that $\dim(G_{\text{NLSS}}) = \dim(\chi[G_{\text{NLSS}}]) \leq \dim(G_{\text{NLSS}}^\prime)$, which implies $\mc B$ is minimal.
\end{proof}

As this proof is quite complicated, it is worth making a general observation. Essentially the proof constructs a `canonical' mapping between any two non-local surface stabilizer groups $G_{\text{NLSS}}$ and $G_{\text{NLSS}}^\prime$. The map $\chi$ is well-defined regardless of whether or not \ref{s3} is satisfied. \ref{s3} only guarantees the resulting map is injective. The importance of the basis connection lemma in allowing us to construct a well-defined, linear $\chi$, should be evident from the proof. 

 \bibliography{library}
 \end{document}